%% file: arxiv.tex
\documentclass[a4paper,UKenglish,cleveref,autoref,thm-restate]{lipics-v2021}

\hideLIPIcs  

\usepackage{graphicx} 
\usepackage{tikz}
\usepackage{mathtools}
\usepackage{amssymb,amsmath}
\usepackage{xspace}

\renewcommand{\S}{\mathcal{S}}

\usetikzlibrary{shapes.geometric,shapes.misc}
\usetikzlibrary{calc,snakes,fit,backgrounds}
\newcommand{\LCP}{\mathsf{LCP}}
\newcommand{\LCPR}{\mathsf{LCP}_{\mathsf{R}}}
\newcommand{\IPM}{\mathsf{IPM}}
\newcommand{\bird}{{layer}\xspace}
\newcommand{\birds}{{layers}\xspace}

\newcommand{\Pyr}{\mathbf{P}\xspace}
\newcommand{\Sync}{\mathbf{Sync}\xspace}
\newcommand{\flocks}{pyramids\xspace}

\newcommand{\RF}{\mathbf{RegP}\xspace}

\newcommand{\ov}{\mathsf{ov}\xspace}
\newcommand{\type}{\mathsf{type}}

\title{Counting Distinct Square Substrings in Sublinear Time}



\category{} 

\relatedversion{} 

\nolinenumbers 

\EventEditors{Pawe\l{} Gawrychowski, Filip Mazowiecki, and Micha\l{} Skrzypczak}
\EventNoEds{3}
\EventLongTitle{50th International Symposium on Mathematical Foundations of Computer Science (MFCS 2025)}
\EventShortTitle{MFCS 2025}
\EventAcronym{MFCS}
\EventYear{2025}
\EventDate{August 25--29, 2025}
\EventLocation{Warsaw, Poland}
\EventLogo{}
\SeriesVolume{345}
\ArticleNo{59}

\newcommand{\cO}{\mathcal{O}}
\newcommand{\Oh}{\mathcal{O}}

\newcommand{\Ohtilde}{\mathcal{\tilde{O}}}

\newcommand{\RUNS}{\mathsf{Runs}}

\newcommand{\rot}{\textsf{rot}}
\newcommand{\per}{\textsf{per}}
\newcommand{\floor}[1]{\left\lfloor #1 \right\rfloor}
\newcommand{\ceil}[1]{\left\lceil #1 \right\rceil}

\newtheorem{fact}[theorem]{Fact}

\newcommand{\defproblem}[3]{
\vspace{2mm}
\noindent\fbox{
   \begin{minipage}{0.96\textwidth}
   \textsc{#1}\\
   {\bf{Input:}} #2  \\
   {\bf{Output:}} #3
   \end{minipage}
   }
   \vspace{2mm}
}

\def\pillar{{\tt PILLAR}\xspace}
\def\dd{\mathinner{.\,.}}

\newcommand{\Lroot}{\textsf{Lroot}}
\newcommand{\Lrepr}{\textsf{Lrepr}}
\newcommand{\squares}{\textsf{squares}}
\newcommand{\sLroot}{\textsf{sLroot}}
\newcommand{\subper}{\textsf{subper}}
\newcommand{\fragsquares}{\mathrm {frag\mbox{-}squares}}

\author{Panagiotis Charalampopoulos}{King's College London, UK}{p.charalampopoulos@kcl.ac.uk}{https://orcid.org/0000-0002-6024-1557}{}
\author{Manal Mohamed}{Birkbeck, University of London, UK}{manalabd@gmail.com}{https://orcid.org/0000-0002-1435-5051}{}
\author{Jakub Radoszewski}{University of Warsaw, Poland}{jrad@mimuw.edu.pl}{https://orcid.org/0000-0002-0067-6401}{Supported by the Polish National Science Center, grant no.\ 2022/46/E/{\allowbreak}ST6/\allowbreak00463.}
\author{Wojciech Rytter}{University of Warsaw, Poland}{rytter@mimuw.edu.pl}{https://orcid.org/0000-0002-9162-6724}{}
\author{Tomasz Wale\'n}{University of Warsaw, Poland}{walen@mimuw.edu.pl}{https://orcid.org/0000-0002-7369-3309}{}
\author{Wiktor Zuba}{University of Warsaw, Poland}{w.zuba@mimuw.edu.pl}{https://orcid.org/0000-0002-1988-3507}{}

\authorrunning{P. Charalampopoulos et al.} 
\Copyright{Panagiotis Charalampopoulos, Manal Mohamed, Jakub Radoszewski, Wojciech Rytter, Tomasz Waleń, Wiktor Zuba} 

\begin{document}
\maketitle
\begin{abstract}
We show that the number of distinct squares in a packed string of length $n$ over an alphabet of size $\sigma$ can be computed in $\Oh(n/\log_\sigma n)$ time in the word-RAM model. This paper is the first to introduce a sublinear-time algorithm for counting squares in the packed setting. The packed representation of a string of length $n$ over an alphabet of size $\sigma$ is given as a sequence of $\Oh(n/\log_\sigma n)$ machine words in the word-RAM model (a machine word consists of $\omega \ge \log_2 n$ bits). Previously, it was known how to count distinct squares in $\Oh(n)$ time [Gusfield and Stoye, JCSS 2004], even for a string over an integer alphabet [Crochemore et al., TCS 2014; Bannai et al., CPM 2017; Charalampopoulos et al., SPIRE 2020]. We use the techniques for extracting squares from runs described by Crochemore et al. [TCS 2014]. However, the packed model requires novel approaches.

We need an $\Oh(n/\log_\sigma n)$-sized representation of all long-period runs (runs with period $\Omega(\log_\sigma n)$) which allows for a sublinear-time counting of the---potentially linearly-many---implied squares. The long-period runs with a string period that is periodic itself (called layer runs) are an obstacle, since their number can be $\Omega(n)$. The number of all other long-period runs is $\Oh(n/\log_\sigma n)$ and we can construct an implicit representation of all long-period runs in $\Oh(n/\log_\sigma n)$ time by leveraging the insights of Amir et al. [ESA 2019]. We count squares in layer runs by exploiting combinatorial properties of pyramidally-shaped groups of layer runs. Another difficulty lies in computing the locations of Lyndon roots of runs in packed strings, which is needed for grouping runs that may generate equal squares. To overcome this difficulty, we introduce sparse-Lyndon roots which are based on string synchronizers [Kempa and Kociumaka, STOC 2019].
\end{abstract}

\section{Introduction}
We consider a problem of counting distinct squares (and more generally, powers) in a string.
Such problems are important not
only from a purely theoretical point of view, but are also 
relevant in some applications in bioinformatics (see the book~\cite{DBLP:books/cu/Gusfield1997}).
Strings of the form $X^2=XX$, for a non-empty string $X$, called \emph{squares} (or tandem repeats),  are  the most
natural type of repetition.

A fundamental algorithmic problem related to squares is checking
if a given string of length~$n$ is \emph{square-free}, that is, if it avoids square substrings.
Thue's construction of an infinite ternary square-free string~\cite{Thue} can be viewed as the beginning of combinatorics on words. 
The first $\Oh(n \log n)$-time algorithm for checking square-freeness was given by Main and Lorentz~\cite{DBLP:journals/jal/MainL84}. An $\Oh(n)$-time algorithm for this problem, for the case of a constant-sized alphabet, was proposed by Crochemore~\cite{DBLP:journals/tcs/Crochemore86}. Subsequently, $\Oh(n)$-time algorithms for square-freeness over an integer alphabet and over a general ordered alphabet follow from Kolpakov and Kucherov's~\cite{DBLP:conf/focs/KolpakovK99} and Ellert and Fischer's~\cite{DBLP:conf/icalp/Ellert021} algorithms for computing runs under these assumptions, respectively.
Most recently, Ellert, Gawrychowski, and Gourdel~\cite{DBLP:conf/soda/EllertGG23} obtained an $\Oh(n \log \sigma)$-time algorithm for testing square-freeness of a string over a general unordered alphabet of size $\sigma$; they also showed that the algorithm is optimal under these assumptions. Square-freeness was also studied in on-line~\cite{DBLP:journals/tcs/HongC08,DBLP:conf/cpm/Kosolobov15}, parallel~\cite{DBLP:journals/algorithmica/Apostolico92,DBLP:journals/siamcomp/ApostolicoB96,DBLP:journals/ipl/CrochemoreR91} and dynamic~\cite{DBLP:conf/esa/AmirBCK19} settings.

A  much more challenging problem than testing square-freeness, that has received significant attention, is computing the number of \emph{distinct} square substrings of a given string.
Fraenkel and Simpson were the first to show that a string of length $n$ contains $\Oh(n)$ dictinct squares~\cite{DBLP:journals/jct/FraenkelS98}. Brlek and Li very recently, using arguments from linear algebra and graph theory, improved the $2n$ upper bound of Fraenkel and Simpson to just $n$; see~\cite{brlek2022numbersquaresfiniteword,DBLP:conf/cwords/BrlekL23}.

Linear-time algorithms for counting distinct square substrings were proposed by Gusfield and Stoye~\cite{DBLP:journals/jcss/GusfieldS04}, Crochemore et al.~\cite{DBLP:journals/tcs/CrochemoreIKRRW14}, Bannai, Inenaga, and K{\"{o}}ppl~\cite{DBLP:conf/cpm/BannaiIK17}, and Charalampopoulos et al.~\cite{DBLP:conf/spire/Charalampopoulos20}; notably, the last three results work for a string over an integer (generally, linearly sortable) alphabet.
As already mentioned, testing square-freeness is a simpler problem than counting distinct squares. In particular,
for a general (ordered) alphabet, element distinctness can be reduced in linear time to counting squares\footnote{As for the reduction, a sequence $a_1,\ldots,a_n$ contains a repeating element, if and only if, string $a_1^2b_1a_2^2b_2\cdots a_n^2b_n$, for a square-free string $b_1b_2\dots b_n$ (say, a prefix of the infinite ternary square free string~\cite{Thue}), contains less than $n$ distinct square substrings.}, and the latter problem is hard: it requires $\Omega(n\log n)$ time in the comparison model \cite{DBLP:conf/stoc/Ben-Or83}.

In the word-RAM model of computation with word size $\Theta(\log n)$, we may store up to
$\Omega(\log_{\sigma} n)$ string characters in a single machine word, where $\sigma$ is the size of the alphabet.
The \emph{packed representation} of a length-$n$ string $S$ over an integer alphabet
$[0 \dd \sigma)$ is a sequence of $\Oh(n/\log_{\sigma} n)$ integers, each encoding a fragment of $S$ of length $\Oh(\log_{\sigma} n)$.

A recent line of work has yielded $o(n)$-time solutions for several basic stringology problems in the setting where the input string(s) is/are given in packed form (the \emph{packed setting}).
These include pattern matching \cite{DBLP:journals/tcs/Ben-KikiBBGGW14} and indexing~\cite{DBLP:conf/stoc/KempaK19,DBLP:conf/cpm/MunroNN20}, computing the LZ factorization and BWT~\cite{DBLP:conf/spire/Ellert23,DBLP:conf/soda/Kempa19,DBLP:conf/stoc/KempaK19,DBLP:conf/focs/KempaK24}, the longest common substring~\cite{DBLP:conf/esa/Charalampopoulos21}, the longest palindromic substring~\cite{DBLP:conf/cpm/Charalampopoulos22}, the Lyndon array \cite{DBLP:conf/esa/BannaiE23}, and covers of a string~\cite{DBLP:conf/spire/RadoszewskiZ24}.
While for some of the discussed problems, such as pattern matching, optimal $\cO(n/ \log_\sigma n)$-time algorithms exist, for several others, such as BWT construction, the best known algorithms run in time $\cO(n\sqrt{\log n} / \log_\sigma n)$.
A~recent work of Kempa and Kociumaka~\cite{kk2025packed} shed light on the source of difficulty of several stringology problems
for which the state-of-the-art solutions take $\cO(n\sqrt{\log n} / \log_\sigma n)$ time.

We are the first to study the problem of counting squares in the packed setting. The problem is formally defined as follows.

\defproblem{Packed Counting Of Distinct Squares}{A string $T$ of length $n$ over alphabet $[0\dd \sigma)$ given in a packed representation.}{$|\squares(T)|$, where $\squares(T)$ denotes the set of all squares that are equal to some substring of $T$.}

\begin{example}\label{ex:triv'}
Consider string $T = (\texttt{ab})^{1000} (\texttt{ba})^{1000}$; 
see \cref{fig:enter-label} for a comparison.  This string contains 2000 distinct squares. We have
\[\squares(T)\,=\,\{\texttt{b} (\texttt{ab})^i \texttt{b} (\texttt{ab})^i:\, 0\le i\le 999\}\cup 
\{(\texttt{ab})^{2i}:\, 1\le i\le 500\}\cup \{(\texttt{ba})^{2i}:\, 1\le i\le 500\}.\] 
\end{example}

We settle the time complexity of the \textsc{Packed Counting Of Distinct Squares} problem, classifiying it as one of the elementary stringology problems that admit an $\cO(n/ \log_\sigma n)$-time solution in the packed setting.

\begin{restatable}{theorem}{mainthm}\label{maintheorem}
The \textsc{Packed Counting Of Distinct Squares} problem can be solved in $\cO(n / \log_\sigma n)$ time.
\end{restatable}

Moreover, our algorithm can report $k$ distinct squares, for any $k$ between 0 and the actual number of distinct squares in the string, in $\Oh(n/\log_{\sigma}n+k)$ time.
Our algorithm generalizes readily to powers with higher exponent: for any integer $t \ge 2$, a string of length $n$ contains at most $n/(t-1)$ powers with exponent $t$~\cite{DBLP:journals/combinatorics/LiPR24} and we can compute the actual number of those in a packed string in $\Oh(n/\log_{\sigma}n)$ time.

\subparagraph*{Other related work.}

A string of length $n$ contains $\Oh(n \log n)$ substrings that are primitively rooted squares and they can all be computed in $\Oh(n \log n)$ time; see~\cite{DBLP:journals/ipl/Crochemore81,DBLP:journals/algorithmica/CrochemoreR95,DBLP:journals/tcs/StoyeG02}. The same representation is computed by Apostolico and Breslauer's parallel algorithm in \cite{DBLP:journals/siamcomp/ApostolicoB96}.
We note that these algorithms compute \emph{all occurrences} of primitively rooted squares and are not concerned with whether any two computed substrings correspond to the same square.

\subsection*{Technical Overview}
Let $T$ be a string of length $n$ over alphabet $[0 \dd \sigma)$. 
It was already observed by Crochemore et al.~\cite{DBLP:journals/tcs/CrochemoreIKRRW14} that distinct squares in $T$ can be computed from runs (maximal periodic fragments). This is because a square $U^2$ can be extended to a unique run with the same period as the primitive root of $U^2$ and a string of length $n$ contains $\Oh(n)$ runs that can be computed in $\Oh(n)$ time~\cite{DBLP:conf/focs/KolpakovK99,DBLP:journals/siamcomp/BannaiIINTT17,DBLP:conf/icalp/Ellert021}. Amir et al.~\cite{DBLP:conf/esa/AmirBCK19} (see also \cite{Panos}) proposed an algorithm that efficiently maintains a representation of runs in a dynamic string.
We note that their algorithm works in the \pillar model of Charalampopoulos, Kociumaka, and Wellnitz~\cite{DBLP:conf/focs/Charalampopoulos20} and it can be used to compute a representation of all runs in $T$ with period at least $\log_\sigma n$ in $\Oh(n/\log_\sigma n)$ time.
All the remaining runs in $T$ either fit in a machine word or are so-called $\tau$-runs (see~\cite{DBLP:conf/stoc/KempaK19}); the number of the latter is $\Oh(n/\log_\sigma n)$. As a result, a representation of all runs in $T$ can be computed in $\Oh(n/\log_\sigma n)$ time and space. This representation can be used to compute the longest square in a string in $\Oh(n/\log_\sigma n)$ time in a straightforward way.

Computing the number of distinct squares  from this representation is much more challenging.
The first obstacle  towards achieving this goal is the difficulty in grouping the runs with respect to their Lyndon roots in packed strings.
We overcome this obstacle by introducing a version of Lyndon roots more suitable for the packed model, called here \emph{sparse} Lyndon roots. Positions of such 
nonstandard roots are based on synchronizing sets of
positions (see Kempa and Kociumaka \cite{DBLP:conf/stoc/KempaK19}).

Let us call squares $U^2$ whose root $U$ is both primitive and highly periodic (here: containing at least 4 occurrences of the period) \emph{special}, while remaining squares are called {\it plain}.
Plain squares can be efficiently counted by combining tabulation with the approach of \cite{DBLP:journals/tcs/CrochemoreIKRRW14} applied on a selected $\cO(n / \log_\sigma n)$-sized subset of the runs of $T$, after grouping these runs by their Lyndon roots (for small-period runs) or sparse Lyndon roots (for long-period runs).
Counting special squares is significantly more challenging.
We tackle this problem by processing certain families of runs.
The runs corresponding to special squares with 
large period (larger than $\log_\sigma n$) are called here {\it layer-runs}. 
The crucial point is that these layer-runs
can be grouped in $\cO(n / \log_\sigma n)$ groups called
``pyramids'', though they can contain together $\Omega(n)$ layer-runs. 
These ``pyramids'' are very regular and counting squares in them can be done 
in batches in sublinear time.
Let us provide a trivial yet illustrative example.

\begin{example}\label{ex:triv}
Consider string $S = (\texttt{ab})^{m} (\texttt{ba})^{m}$.
For each $i \in [0 \dd m)$, we have a square of the form $\texttt{b} (\texttt{ab})^i \texttt{b} (\texttt{ab})^i$ that occurs at (0-based) position $2m - 1 - 2i$; this square is primitively rooted and, in fact, $S[2m - 1 - 2i \dd 2m  + 2i]$ is a run.
These are the only primitively rooted squares in $S$ other than primitively rooted squares $\texttt{abab}$ and $\texttt{baba}$; 
see \cref{fig:enter-label}. 

\begin{figure}[h]
    \centering
    \input{_fig_pyramid}
    \caption{The pyramidal-shaped structure of six special squares contained in $(\mathtt{ab})^8(\mathtt{ba})^8$.}
    \label{fig:enter-label}
\end{figure}

Using the approach of Crochemore et al.~\cite{DBLP:journals/tcs/CrochemoreIKRRW14}, we can count all plain squares in the string from \cref{ex:triv} by processing a constant number of runs: $(\texttt{ab})^{m}$, $(\texttt{ba})^{m}$, $\texttt{bb}$, and $(\texttt{b}(\texttt{ab})^i)^{2}$ for $i \in \{1,2,3\}$.
However, there are $\Theta(n)$ runs of the form $\texttt{b} (\texttt{ab})^i \texttt{b} (\texttt{ab})^i$ for $i \geq 4$,
each corresponding to a special square, and we clearly cannot afford to iterate over these runs in $o(n)$ time.
All such special squares in our example have their first half in run $(\texttt{ab})^{m}$ and their second half in run $(\texttt{ba})^{m}$, and that these two runs have the same Lyndon root (\texttt{ab}).
\end{example}

We employ an $\cO(n/\log_\sigma n)$-sized representation of all runs that generate special squares.
As presented intuitively in \cref{ex:triv}, the representation relies on $\cO(n/\log_\sigma n)$ pairs $T[a \dd c]$ and $T[b \dd d]$ of runs that have the same Lyndon root $\lambda$ and satisfy $b \in (c-|\lambda|+1 \dd c+1]$.
Counting special squares from said representation is reduced to appropriately grouping the computed pairs of runs.

\section{Preliminaries}
For a string $S$, its positions are numbered as $S[0],\ldots,S[|S|-1]$. If $|S|=0$, $S$ is called the empty string. A string consisting of characters $S[i],S[i+1],\ldots,S[j]$ is a substring of $S$.  A fragment of $S$ is a positioned substring; that is, a fragment $S[i \dd j]$ represents the substring $S[i],S[i+1],\ldots,S[j]$. We also denote $S[i \dd j]=S[i \dd j+1)$.
Where possible, we treat fragments and substrings as equivalent.
For two fragments $F=S[a \dd b]$ and $F'=S[a' \dd b']$, we write $F \subseteq F'$ if $[a \dd b] \subseteq [a' \dd b']$. Two fragments $S[a \dd b]$ and $S[a' \dd b']$ are called \emph{neighboring} if $[a-1 \dd b+1] \cap [a' \dd b'] \ne \emptyset$. For two neighboring fragments $F=S[a \dd b]$ and $F'=S[a' \dd b']$, by $F \cup F'$ we denote the fragment $S[\min(a,a') \dd \max(b,b')]$ and by $F \cap F'$ we denote the fragment $S[\max(a,a') \dd \min(b,b')]$.

We say that a positive integer $p$ is a period of string $S$ if $p \le |S|$ and $S[i]=S[i+p]$ for all $i \in [0 \dd |S|-p)$; equivalently, $S[0 \dd |S|-p)=S[p \dd |S|)$. Fine and Wilf's periodicity lemma~\cite{fine1965uniqueness} asserts that if a string of length~$n$ has periods $p$ and $q$ such that $p+q \le n$, then the string has period $\gcd(p,q)$.
By $\per(S)$ we denote the smallest period of $S$ to which we refer as \emph{the period} of $S$.

A non-empty string $S$ is called primitive if the equality $S=U^t$ for a positive integer $t$ implies that $t=1$. By $\rot_c(S)$ we denote a cyclic rotation of the string $S$, obtained by moving the $c$ first characters of $S$ to its end. A string is called a \emph{Lyndon string} if it is primitive and lexicographically minimal in the class of its cyclic rotations.

We say that a string $S$ is \emph{periodic} if $\per(S) \le \frac12 |S|$ and \emph{highly periodic} if $\per(S) \le \frac14 |S|$.
The \emph{Lyndon root} of a periodic string~$S$, denoted by $\Lroot(S)$, is the lexicographically smallest rotation of $S[0\dd \per(S))$.
For example,
    $\Lroot((\mathtt{abaaa})^3\, \mathtt{aba})\, =\,\mathtt{aaaab}$.

\begin{definition}\label{def:lrep}
The \emph{Lyndon representation} of a periodic string $U$ is a quadruple $\Lrepr(U)=(\lambda,e,\alpha,\beta)$ such that:

\begin{itemize}
\item $\lambda=\Lroot(U)$, and
\item $U=P\lambda^e S$ with $|P|= \alpha < |\lambda|$ and $|S|=\beta<|\lambda|$. ($P$ and/or $S$ can be the empty string.)
\end{itemize}
\end{definition}

\begin{example} We have $\Lrepr(U)\,=\,(\mathtt{aaaab},3,2,1)$ for
 \[U\,=\,
 (\mathtt{abaaa})^3\,\mathtt{aba}
 \,=\,\mathtt{abaaa\, abaaa\, abaaa\, aba}
 \,=\, \mathtt{ab}\, (\mathtt{aaaab})^3\, \mathtt{a}. \] 
\end{example}

A \emph{run} in a string $T$ is a fragment $F=T[a \dd b]$ that is periodic, that is, $p:=\per(F) \le \frac12 |F|$, and inclusion-maximal, that is,
\begin{itemize}
\item $a=0$ or $T[a-1] \ne T[a-1+p]$ and
\item $b=|T|-1$ or $T[b+1] \ne T[b+1-p]$.
\end{itemize}
We denote the set of runs in $T$ by $\RUNS(T)$.
A string of length $n$ contains at most $n$ runs and they can be computed in $\Oh(n)$ time~\cite{DBLP:journals/siamcomp/BannaiIINTT17}, even if the string is over an arbitrary ordered alphabet~\cite{DBLP:conf/icalp/Ellert021}. 

The \emph{Lyndon position} of a run $R=T[a \dd b]$ with Lyndon root $\lambda$ is the unique position $i \in [a \dd a+\per(R))$ such that $T[i \dd i+\per(R)) = \lambda$.

A square is a string of the form $X^2=XX$ for a non-empty string $X$. A square $X^2$ is called \emph{primitively rooted} if $X$ is primitive and \emph{non-primitively-rooted} otherwise.

We say that a square $X^2$ is \emph{generated} by a periodic string $U$ if $X^2$ is contained in $U$ and $\per(X^2)=\per(U)$. By the periodicity lemma, in this case $|X|$ is a multiple of $\per(U)$.
We denote by 
$\fragsquares(U)$
the set of squares generated by a periodic string $U$.

\begin{example}
For the underlined run $R=T[1\dd 13]$ with period 2 in $T=\mathtt{c\underline{abababababab}d}$, we have
$\fragsquares(R)=\{(\mathtt{ab})^2,(\mathtt{ba})^2,(\mathtt{abab})^2,(\mathtt{baba})^2,(\mathtt{ababab})^2\}$.
\end{example}
For a set $\mathcal{X}$ of periodic fragments of $T$ we denote \[\fragsquares(\mathcal{X})\,=\, \bigcup_{U\in \mathcal{X}}\, \fragsquares(U).\]
The following observation states that each square in a string is generated by a run.

\begin{observation}[\cite{DBLP:journals/tcs/CrochemoreIKRRW14}]\label{obs:square_run}
$\squares(T)\,=\,  \fragsquares(\RUNS(T)).$
\end{observation}

Unfortunately, from the point of view of counting, 
the same square string can be generated by many runs.
Dealing with squares whose first half is both primitive and (highly) periodic is the most challenging.
The next section is devoted to runs that generate such squares.

The next fact follows by using radix sort (with bucket sort).

\begin{fact}\label{fact:bsort}
A list of $\Oh(n/\log_\sigma n)$ $k$-tuples of integers in $[0 \dd n)$, for any constant $k>0$, can be sorted lexicographically in a stable manner in $\Oh(n/\log_\sigma n)$ time.
\end{fact}
\begin{proof}
We treat each number $x$ from $[1\dd n]$ in the base-$m$ number system for $m=\lceil\sqrt{n}\rceil$, that is, $x$ is a 2-tuple of integers $(u,v)$ in $[0\dd m)$ such that $x=u\cdot m+v$. Thus each $k$-tuple is transformed into a $2k$-tuple with components in $[0 \dd m)$. We then apply radix sort to the resulting tuples (using bucket sort on $m$ buckets in a stable way).
\end{proof}

\section{Pyramids of Runs}
In this section we describe the structure of runs that generate special squares.
\begin{definition}[Subperiodic strings]\label{def:subper}
For a periodic string $U$ we define 
\[\subper(U)\,=\, \min \,\{\, \per(X)\,:\, X^2  \in \fragsquares(U)\,\}.\] 
A periodic string $U$ is called \emph{subperiodic} if $\subper(U)\le \per(U)/4$.
\end{definition}

\begin{example}
The string $U\,=\, \mathtt{ab}\,(\mathtt{babababab})^5\,
\mathtt{ba}$ with period 5 generates a subperiodic square $(\mathtt{babababab})^2$ of length $2\cdot \per(U)$. Thus $U$ is subperiodic, and $\subper(U)=2$. Observe that all the remaining squares generated by $U$, in particular, $(\mathtt{ababababb})^2$ and $(\mathtt{babababab})^4$, are not subperiodic.
\end{example}

A special square is a square which is primitively rooted and subperiodic. All other squares are called plain.

For two neighboring runs $F,F'$ with equal period $p$ in $T$, we have $|F\cap F'| \leq p-1$~\cite{DBLP:conf/focs/KolpakovK99}.
Two such runs can induce a collection of runs with subperiod $p$.
We formalize this structure in the following definition and provide illustrations in \cref{fig:flock,fig:flock2}.

\begin{definition}\label{pyramid} Let $F$ and $F'$ be neighboring runs in $T$ with period $p$ and equal Lyndon roots. A \emph{pyramid} $\Pyr(F,F')$ of runs is the set
\[\{R\,:\, R \ \mbox{is a subperiodic run},\  \subper(R)=p,\
R \cap (F \cup F') \text{ is periodic with period } \per(R) \}.\]
If $R\in \Pyr(F,F')$, run $R$ is called a \emph{layer-run} (or a \emph{layer} for brevity). 
\end{definition}

\begin{remark}\label{rem:nonemptyfs}
We show that all layers in $\Pyr(F,F')$ are contained in $F \cup F'$
except possibly the longest layer (for example the red layer in \cref{fig:flock2}).
\end{remark}

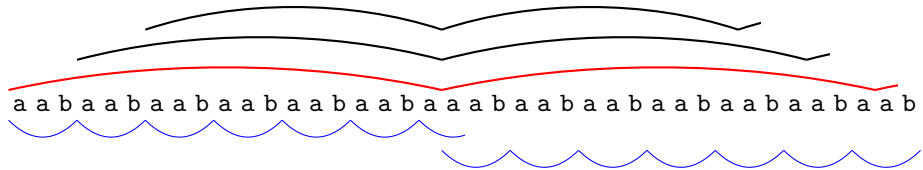
\begin{figure}[h]
\centering
\begin{tikzpicture}
\foreach \c [count=\i] in {a,a,b,a,a,b,a,a,b,a,a,b,a,a,b,a,a,b,a,a,a,b,a,a,b,a,a,b,a,a,b,a,a,b,a,a,b,a,a,b}{
\draw (\i*0.3,0) node[above] {\texttt{\c}};
}
\begin{scope}[xshift=1.65cm,yshift=0.8cm]
\clip (0,0.3) rectangle (8.4,1);
\foreach \dx in {0.3,4.2,8.1}{
\draw[black,thick,xshift=\dx cm] (0,0.4) .. controls (4*0.3,0.8) and (9*0.3,0.8) .. (13*0.3,0.4);
}
\end{scope}
\begin{scope}[xshift=1.65cm,yshift=0.4cm]
\clip (-1,0.3) rectangle (9.3,1);
\foreach \dx in {-0.6,4.2,9.0}{
\draw[black,thick,xshift=\dx cm] (0,0.4) .. controls (5*0.3,0.8) and (11*0.3,0.8) .. (16*0.3,0.4);
}
\end{scope}
\begin{scope}[xshift=1.65cm,yshift=0.0cm]
\clip (-2,0.3) rectangle (10.2,1);
\foreach \dx in {-1.5,4.2,9.9}{
\draw[red,thick,xshift=\dx cm] (0,0.4) .. controls (6*0.3,0.8) and (13*0.3,0.8) .. (19*0.3,0.4);
}
\end{scope}
\begin{scope}[xshift=-0.15cm]
\clip (0,-0.5) rectangle (6.3,0.1);
\foreach \dx in {0.3,1.2,2.1,3.0,3.9,4.8,5.7}{
\draw[blue,xshift=\dx cm] (0,0) .. controls (0.3,-0.3) and (0.6,-0.3) .. (0.9,0);
}
\end{scope}
\begin{scope}[xshift=5.55cm,yshift=-0.4cm]
\clip (0,-0.5) rectangle (6.6,0.1);
\foreach \dx in {0.3,1.2,2.1,3.0,3.9,4.8,5.7}{
\draw[blue,xshift=\dx cm] (0,0) .. controls (0.3,-0.3) and (0.6,-0.3) .. (0.9,0);
}
\end{scope}
\end{tikzpicture}
    \caption{
    The runs $F$ and $F'$ with period 3 (at the bottom, in blue) imply a pyramid $\Pyr(F,F')$ containing three layer-runs     $(\texttt{aab})^i\texttt{a}\, (\texttt{aab})^i\texttt{aa}$, for $i \in \{4,5,6\}$ (above). 
    }\label{fig:flock}
\end{figure}

\begin{figure}[h]
\centering
\begin{tikzpicture}
\foreach \c [count=\i] in {\,,a,b,a,b,a,b,a,b,a,b,a,b,a,b,a,a,b,a,b,a,b,a,b,a,b,a,b,a,b,a,a,b,a,b,a,b,a,b,a,b,a,b,a,b,a}{
\draw (\i*0.3,0) node[above] {\texttt{\c}};
}
\begin{scope}[xshift=0.15cm,yshift=0.1cm]
\begin{scope}[yshift=0.3cm]
\foreach \dx in {2.1,4.8}{
\draw[black,thick,xshift=\dx cm,yshift=1.05cm] (0,0.4) .. controls (3*0.3,0.6) and (6*0.3,0.6) .. (9*0.3,0.4);
}
\foreach \dx in {1.5,4.8}{
\draw[black,thick,xshift=\dx cm,yshift=0.65cm] (0,0.4) .. controls (4*0.3,0.6) and (7*0.3,0.6) .. (11*0.3,0.4);
}
\foreach \dx in {0.9,4.8}{
\draw[black,thick,xshift=\dx cm,yshift=0.25cm] (0,0.4) .. controls (5*0.3,0.6) and (8*0.3,0.6) .. (13*0.3,0.4);
}

\foreach \dx in {6.6,9.3}{
\draw[violet,thick,xshift=\dx cm,yshift=0.84cm] (0,0.4) .. controls (3*0.3,0.6) and (6*0.3,0.6) .. (9*0.3,0.4);
}
\foreach \dx in {6.0,9.3}{
\draw[violet,thick,xshift=\dx cm,yshift=0.44cm] (0,0.4) .. controls (4*0.3,0.6) and (7*0.3,0.6) .. (11*0.3,0.4);
}
\foreach \dx in {5.4,9.3}{
\draw[violet,thick,xshift=\dx cm,yshift=0.04cm] (0,0.4) .. controls (5*0.3,0.6) and (8*0.3,0.6) .. (13*0.3,0.4);
}
\end{scope}
\begin{scope}
\foreach \dx in {0.3,4.8,9.3}{
\draw[red,thick,xshift=\dx cm,yshift=-0.05cm] (0,0.4) .. controls (5*0.3,0.7) and (10*0.3,0.7) .. (15*0.3,0.4);
}
\end{scope}
\end{scope}
\begin{scope}[xshift=0.3cm]
\begin{scope}[xshift=-0.15cm]
\clip (0,-0.5) rectangle (4.8,0.1);
\foreach \dx in {0.3,0.9,1.5,2.1,2.7,3.3,3.9,4.5}{
\draw[blue,xshift=\dx cm] (0,0) .. controls (0.2,-0.2) and (0.4,-0.2) .. (0.6,0);
}
\end{scope}
\begin{scope}[xshift=4.35cm]
\clip (0,-0.5) rectangle (4.8,0.1);
\foreach \dx in {0.3,0.9,1.5,2.1,2.7,3.3,3.9,4.5}{
\draw[blue,xshift=\dx cm] (0,0) .. controls (0.2,-0.2) and (0.4,-0.2) .. (0.6,0);
}
\end{scope}
\begin{scope}[xshift=8.85cm]
\clip (0,-0.5) rectangle (4.8,0.1);
\foreach \dx in {0.3,0.9,1.5,2.1,2.7,3.3,3.9,4.5}{
\draw[blue,xshift=\dx cm] (0,0) .. controls (0.2,-0.2) and (0.4,-0.2) .. (0.6,0);
}
\end{scope}
\end{scope}
\end{tikzpicture}
    \caption{
    The subsequent runs $F,F',F''$ with period 2 
    each corresponding to string 
    $(\mathtt{ab})^7\mathtt{a}$ 
    (at the bottom, blue) imply pyramids $\Pyr(F,F')$ and $\Pyr(F',F'')$ (above).
    The longest
     layer (in red) corresponding to string 
    $((\mathtt{ab})^7\mathtt{a})^3$ is common to both pyramids.
    }\label{fig:flock2}
\end{figure}
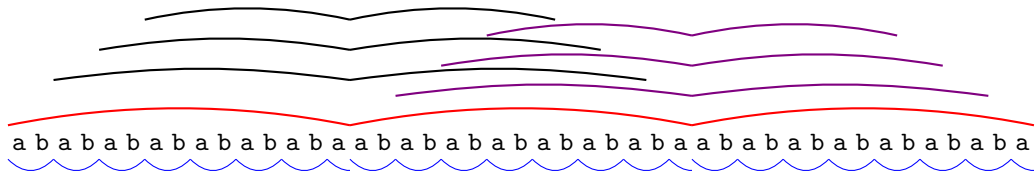

\begin{lemma}\label{lem:util}
    If $U^2 \in \fragsquares(R \cap (F \cup F'))$ for some layer $R$ in pyramid $\Pyr(F,F')$, then the first half of $U^2$ is contained in $F$ and the other 
    half in $F'$.
\end{lemma}
\begin{proof}
Let $p=\per(F)=\per(F')$. Layer $R$ is subperiodic with $\subper(R)=p$ and thus must have period at least $4p$, so $|U| \ge 4p$.

First we show that $U^2$ cannot be a fragment of $F$ (or of $F'$). Indeed, this would mean that $U^2$ has period $p$ as well as period $|U|$. By the periodicity lemma applied to $U^2$, $p$ would divide~$|U|$. Consequently, $p$ would be a period of $R$, a contradiction.

Let us now show, by contradiction, that each half of $U^2$ is contained in $F$ or in $F'$. 
Suppose that this is not the case.
One of the two halves is fully contained in one of $F$ and $F'$ and hence has period $p$, while the other half contains a position that is in $F$ but not in~$F'$ and a position that is in $F'$ but not in $F$, and hence has period greater than $p$ due to the maximality of runs $F$ and $F'$. We have thus obtained a contradiction.
\end{proof}

Next we obtain a combinatorial characterization of all runs in a pyramid.

\begin{definition} 
A layer with a maximal period in a pyramid is called a \emph{max-layer}.
We denote by ${\RF}(F,F')$ the set of layer-runs in $\Pyr(F,F')$ without the max-layer.
The elements of ${\RF}(F,F')$ are called \emph{regular} \birds.
\end{definition}

\begin{example}
In \cref{fig:flock}, there are two regular layers and one max-layer. In \cref{fig:flock2}, the first pyramid contains three regular layers while the second pyramid contains two; there is one max-layer that is common to all the pyramids.
\end{example}

Consider a pyramid $\Pyr(F,F')$ and let $p = \per(F)$.
A \emph{canonical representation} of pyramid $\Pyr(F,F')$ consists of the (endpoints of) runs $F$ and $F'$, its max-layer, and
sequences specifying the starting positions, ending positions, and periods of its regular layers. In a canonical representation, the ending positions of regular layers form an arithmetic progression with difference $p$, whereas the starting positions form an arithmetic progression with difference~$-p$.
Moreover, the periods of all regular layers form an arithmetic progression with difference $p$. The lemma below is proved in the following \cref{lem:contained,lem:maxlayer}.

\begin{lemma}\label{cor:canonrep}
Any non-empty pyramid admits a canonical representation.
\end{lemma}

    We use the following notation for a fixed pyramid $\Pyr(F,F')$.
    Let $F=T[a\dd b]$, $F'=T[a'\dd b']$, assuming without loss of generality that $a < a'$, and $p= \per(F) = \per(F')$.
    Let the Lyndon positions of $F$ and $F'$ be $\ell$ and $\ell'$, respectively, and define $\delta := (\ell' - \ell) \bmod p$.
    For each $k\in \mathbb{Z}$, we denote $a'_k := a' - k\cdot p - \delta$ and $b_k = b+k\cdot p +\delta$.

    \begin{lemma}\label{lem:contained}
    The set $\mathcal{R} := \{T[x \dd y] \in \Pyr(F,F'): x,y \in (a \dd b')\}$ is equal to
    \[\mathcal{K} := \{T[a'_k \dd b_k] : k \in K \}, \text{ where }K = \{ k \in \mathbb{Z}\ :\ k \geq 4,\ a'_k > a,\ b_k <b' \}.\]
    For each $k \in K$, the period of run $T[a'_k \dd b_k]$ is $k\cdot p + \delta$.
    \end{lemma}
    \begin{proof}
    First, let us argue that $\per(T[a'_k \dd b_k])=k\cdot p + \delta$ for each $k \geq 2$.
    We note that $T[a'_k \dd b] = T[a' \dd b_k]$ (by the definition of $\delta$ and the fact that the two strings are contained in $F$ and $F'$, respectively), so $a'_k-a'$ is a period of $T[a'_k \dd b_k]$ by definition,
    and $a'_k - a' = k\cdot p + \delta$.
    Observe that $T[a'_k \dd b]$ has period $p$ and hence it cannot have an occurrence starting before position $a'$ and ending after position $b$---as all fragments satisfying these conditions do not have period $p$ by the maximality of $F$ and~$F'$.
    Thus, $T[a'_k \dd b_k]$ does not have any period smaller than $k\cdot p + \delta$.
    
    \textbf{$\mathcal{K} \subseteq \mathcal{R}$:}
        Let $R=T[a'_k \dd b_k]$.
        We have
        \[|R|\,=\,|R \cap (F \cup F')|\,=\,b_k-a'_k+1\,=\,b-a'+1+2kp+2\delta\,\ge\, 2kp+2\delta\,=\,2\cdot\per(R)\]
        as $F$ and $F'$ are neighboring.
        In particular, $R$ is periodic.
        Let us show that $R$ is a maximal fragment with period $\per(R)$. As $a'_k>a$, the periodicities of $F$ and $R$ and left maximality of~$F'$ imply that
        \[T[a'_k-1] = T[a'_k-1+p] = T[a'_k-1+p+\per(R)] \ne T[a'_k-1+\per(R)] = T[a'-1],\]
        which shows the left maximality of $R$.
        A symmetric argument yields right maximality.
        Hence,~$R$ is indeed a run.
        
        The prefix $T[a'_k \dd b]$ of $R$ has length $|R|-\per(R) \ge \per(R)=kp+\delta \ge 2p$ and period $p$, so $R$ is subperiodic and $\subper(R)\le p$. Moreover, $\subper(R)$ cannot be smaller than $p$, as then there would be a run with period smaller than $p$ overlapping one of runs $F,F'$ on at least $kp+\delta \ge 4p$ positions, which is impossible due to the periodicity lemma. 
        
    \textbf{$\mathcal{R} \subseteq \mathcal{K}$:}
        Let us fix some $R = T[\alpha \dd \beta] \in \mathcal{R}$.
        As $R$ is subperiodic, $\per(R)\geq 4p$, and by the definition of $\mathcal{R}$, $\alpha>a$ and $\beta<b'$.
        Consider a square $T[x \dd x + 2\cdot\per(R)) \in \fragsquares(R)$.
        By \cref{lem:util},
        $T[x \dd x + \per(R)) \subseteq F$
        and $T[x+\per(R) \dd x + 2\cdot\per(R)) \subseteq F'$.
        Since primitive strings do not match non-trivial rotations of themselves,
        we have that $T[x \dd x+p)$ occurs only at positions $y$ of $T$ contained in $F'$ such that $x - \ell \equiv y - \ell' \pmod p$.
        This implies that $\per(R)$ has to be equivalent to $y-x \equiv \ell' - \ell \equiv \delta \pmod p$.
        Then, for some $k\in \mathbb{Z}$, $\per(R) = k\cdot p + \delta$. By the definition of $\mathcal{R}$, $k \in K$.
        Finally, there can be at most one run with period $\per(R)$ that contains at least $\per(R)$ positions of $F$ and $F'$ and we have shown the existence of such a run in $\mathcal{K}$, so $R \in \mathcal{K}$.
    \end{proof}

    \begin{lemma}\label{lem:maxlayer}
        For the set $\mathcal{R}$ defined in \cref{lem:contained}, there is at most one run $R \in \Pyr(F,F') \setminus \mathcal{R}$.
    \end{lemma}
    \begin{proof}
        Consider the case when there exists a run $R \in \Pyr(F,F')$ containing position $a$.
        We then know that $R \cap (F \cup F')$ has subperiod $p$ and is periodic with period $\per(R)$.
        Consider the square $T[a \dd a+2\cdot \per(R)) \in \fragsquares(R \cap (F \cup F'))$.
        By \cref{lem:util}, we know that the first half of this square is contained in $F$, while the second half is contained in $F'$.
        Hence, the square has subperiod $p$.
        This means that $\per(R) \in [a'-a \dd b-a]$, which is an interval of length at most $p-1$.
        Now, $\per(R)$ also has to be equivalent to $\delta \pmod p$, so there is a single possible value for it, say $y$.
        
        Observe that if $b'$ is not contained in $R$, then $F$ is a common prefix of $T[a \dd |T|)$ and $T[a+y \dd |T|)$, which means that
        $a'+|F|-1 \leq a+ y+|F|-1 < b'$ and hence $|F'| > |F|$.
        Now, a run $R' \in \Pyr(F,F')$ containing position $b'$ would have period in $[b'-b \dd b'-a']$ by symmetric arguments to those above.
        Then, we would have $\per(R') \geq b'-b > |F|$, a contradiction to the fact that a square with period $\per(R')$ generated by $R'$ would have its first half in $F$.

        Finally, if our attempt to compute a run containing position $a$ fails, we perform symmetric computations to find a run $R\in \Pyr(F,F')$ that contains position $b'$, if one exists.
    \end{proof}

\section{Computing a Representation of All Runs}
In this section we show that a representation of all runs in a string of length $n$ over an alphabet $[0 \dd \sigma)$ can be computed in $\Oh(n/\log_\sigma n)$ time.

First we show how to compute runs with large periods.
Some of these runs are grouped in pyramids.

Amir et al.~\cite{DBLP:conf/esa/AmirBCK19} showed how to compute squares and runs in a dynamic string.
Their techniques can be interpreted in the so-called \pillar model, introduced by Charalampopoulos, Kociumaka, and Wellnitz~\cite{DBLP:conf/focs/Charalampopoulos20}.
Recent optimal data structures for $\LCP$ queries \cite{DBLP:conf/stoc/KempaK19} and $\IPM$ queries \cite{DBLP:journals/siamcomp/KociumakaRRW24} in the packed setting
imply that any problem on strings of total length $n$ that can be solved in $\cO(f(n))$ time in the \pillar model, can be solved in $\cO(n/ \log_\sigma n + f(n))$ time in the packed setting.
All in all, we obtain the following fact whose proof closely follows~\cite{DBLP:conf/esa/AmirBCK19,Panos}; it is provided in~\cref{app:fctESA}.

\newcommand{\X}{\mathcal{X}}
\newcommand{\Y}{\mathcal{Y}}
\newcommand{\Z}{\mathcal{Z}}
\begin{restatable}[see {\cite{DBLP:conf/esa/AmirBCK19,Panos}}]{fact}{fctESA}\label{fct:ESA19}
Let $T \in [0\dd \sigma)^n$ be a string given in packed form.
For any constant $c>0$, in time $\Oh(n/\log_{\sigma} n)$, we can compute a set $\X$ of runs such that none of them is a regular layer of any pyramid of $T$ and a set $\Y$ of pyramids given by their canonical representations, such that $|\X|,|\Y|=\Oh(n/\log_{\sigma} n)$, and, for $\Z\,:=\,\bigcup_{(F,F')\in \Y}\, \RF(F,F')$, we have that
\begin{itemize}
\item $\X \cup \Z$
is a superset of all runs in $T$ of period at least $c \log_\sigma n$, and
\item $\X \cap \Z = \emptyset$.
\end{itemize}
\end{restatable}

We do not include max-layers in set $\Z$ as they can be common to many pyramids (cf.\ \cref{fig:flock2}).

\begin{remark}
As shown in the proof of \cref{fct:ESA19} (given in the full version) and implicitly in \cite{DBLP:conf/esa/AmirBCK19,Panos}, for any parameter~$q$, the
number of both max-layers with period at least $q$ and non-layer-runs with period at least~$q$ in a length-$n$ string is $\Oh(n/q)$. We note that the number of \emph{all} layer runs with period at least $q$ can be $\Omega(n)$: for any $q\geq 3$,
the string $S$ from \cref{ex:triv} has at least $(n/2 - q - 1)/2$
layer runs with period at least $q$.
\end{remark}

\newcommand{\R}{\mathcal{R}}
\newcommand{\cluster}{\mathbf{Cluster}}
\begin{definition}[Clusters of runs]
For a set of runs $\X$ in $T$ and a set of integers $D$, we define a \emph{cluster of runs}:
\[\cluster(\X,D)\,=\,\{T[a+d \dd b+d]\,:\,T[a \dd b] \in \X,\,d \in D,\,T[a \dd b]=T[a+d \dd b+d]\}.\]
The \emph{size} of a cluster of runs is defined as $|\X|+|D|$.
\end{definition}

In this work, in all considered clusters of runs, we have $0 \in D$.

\begin{example}
The string $S=\mathtt{\#ababaabaab\$}$ contains runs $S[1 \dd 5]=\mathtt{ababa}$, $S[5 \dd 6]=S[8 \dd 9]=\mathtt{aa}$, $S[3 \dd 10]=\mathtt{abaabaab}$. 
Thus the following string 
\[T=\mathtt{\#\textcolor{violet}{ababaabaab}\$\,\#\textcolor{violet}{ababaabaab}\$\,\#\textcolor{violet}{edcbaedcba}\$\,\#\textcolor{violet}{ababaabaab}\$\,}\]
contains a cluster of runs
$\cluster(\,\{T[1 \dd 5],T[5 \dd 6],T[8 \dd 9],T[3 \dd 10]\},\ \{0,12,36\}\,)$.
\end{example}

A \emph{$\tau$-run} $R$ is a run of length at least $3\tau-1$ with period at most $\frac13\tau$.

\begin{lemma}[{\cite[Section 6.1.2]{DBLP:conf/stoc/KempaK19},\cite[Lemma 10]{DBLP:conf/esa/Charalampopoulos21}}]\label{lem:tauruns}
For a positive integer $\tau$, a string $T \in [0\dd \sigma)^n$ contains $\Oh(n/\tau)$ $\tau$-runs. 
Moreover, if $\tau\leq \frac{1}{9} \log_{\sigma} n$, given a packed representation of $T$, we can compute all $\tau$-runs in $T$
in $\Oh(n/\tau)$ time.
Within the same complexity, we can compute the Lyndon position of each $\tau$-run.
\end{lemma}

The next lemma is proved using tabulation.

\begin{restatable}{lemma}{nottauruns}\label{lem:not_tauruns}
Given a string $T$ in packed form and an integer $\tau \leq \frac{1}{9} \log_{\sigma} n$,
we can compute all runs of length smaller than $3\tau-1$ and period at most $\frac13 \tau$, represented as $\Oh(n/\log_\sigma n)$ clusters of runs, in $\Oh(n/\log_\sigma n)$ time.
The sum of lengths of lists $\X$ across all clusters of runs is~$\Ohtilde(n^{7/18})$.
\end{restatable}
\begin{proof}
In a preprocessing phase, for each string $U$ of length at most $3.5\tau$ over alphabet $[0 \dd \sigma)$, we compute all runs with the required properties starting within $U[1 \dd \floor{\frac12\tau}]$ and ending strictly before the last position of~$U$.
They are stored as a list of runs denoted by $\X(U)$.
The number of such distinct strings is $\Oh(n^{7/18})$ and the total length of the computed lists is $\Ohtilde(n^{7/18})$.
The runs in each string can be computed using a linear-time algorithm~\cite{DBLP:conf/icalp/Ellert021}.

For each string $U$ of length at most $3.5\tau$ over alphabet $[0 \dd \sigma)$, we populate an initially empty list $D(U)$ as follows:
for each integer $i \in \{a\cdot \floor{\frac12\tau}\,:\,a \in \mathbb{Z}\} \cap [0 \dd n)$, we insert $i$ to the list $D(T[i \dd \min(i+\floor{3.5\tau},n))$.

Thus the total size of all lists $D(U)$ is $\Oh(n/\tau)$.
In the end, for each string $U$ of length at most $3.5\tau$ over alphabet $[0 \dd \sigma)$, we report a cluster of runs $\cluster(\X(U), D(U))$.
(If any of $\X(U), D(U)$ is empty, the cluster need not be reported.)
The total size of the computed clusters is $\Oh(n/\log_\sigma n)$ and they are computed in $\Oh(n/\log_\sigma n)$ time.
\end{proof}

Putting everything together, we obtain the following proposition. 

\begin{proposition}\label{prp:runs}
A representation of all runs in a string $T \in [0\dd \sigma)^n$ consisting of a disjoint union of $\Oh(n/\log_\sigma n)$ runs, regular layers of \flocks, and clusters of runs, can be computed in $\Oh(n/\log_\sigma n)$ time.
\end{proposition}
\begin{proof}
Let $\tau:=\floor{\frac{1}{9} \log_{\sigma} n}$.
We compute three disjoint sets of runs whose union equals $\RUNS(T)$.
We assume that $T$ starts and ends with a sentinel character so that the characters $T[0]$, $T[|T|-1]$ are unique in $T$.
\begin{itemize}
\item \textbf{Runs of length smaller than $3\tau-1$ and period at most $\frac13 \tau$.}
We compute those runs in $\Oh(n/\log_\sigma n)$ time using \cref{lem:not_tauruns}.
\item \textbf{$\tau$-runs (that is, 
runs of length at least $3\tau-1$ and period at most $\frac13 \tau$).}
The $\tau$-runs are computed in $\Oh(n/\log_\sigma n)$ time using \cref{lem:tauruns}.
\item \textbf{Runs with period greater than $\frac13 \tau$.}
A representation of all runs with period greater than $\tau$, composed of single runs and \flocks, can be computed in $\Oh(n/\log_\sigma n)$ time using~\cref{fct:ESA19}. The representation may return runs with periods at most $\tau$, which need to be removed to avoid double reporting (by trimming arithmetic progressions).\qedhere
\end{itemize}
\end{proof}

\section{Grouping Runs via Lyndon Roots and Sparse-Lyndon Roots}\label{sec:grouping}
The next observation is crucial in counting distinct square substrings of a string.
Let us denote by $\RUNS(T,\lambda)$ and $\squares(T,\lambda)$ the sets of runs and squares in $T$ with Lyndon root $\lambda$.

\begin{observation}[\cite{DBLP:journals/tcs/CrochemoreIKRRW14}]\label{obs:sq_from_runs}
Consider two runs $R$ and $R'$ in a string $T$. 
Then, $\fragsquares(R)\cap \fragsquares(R')\ne \emptyset$ implies that $\Lroot(R)=\Lroot(R')$.
In particular, for any Lyndon string~$\lambda$, we have
\[\squares(T,\lambda)=\bigcup_{R\in \RUNS(T,\lambda)}\, \fragsquares(R).\]
\end{observation}

Crochemore et al.~\cite{DBLP:journals/tcs/CrochemoreIKRRW14} considered all runs in the string in groups consisting of runs with equal Lyndon root. The algorithm for grouping of runs that they used consists of the following three steps:

\begin{enumerate}
\item\label{it1} Computing a Lyndon position for each run.
\item\label{it2} Sorting runs with equal periods in the order of the suffixes starting at Lyndon positions. It is guaranteed that runs from the same group are listed consecutively.
\item\label{it3} Partitioning the sorted list of runs obtained for each period into groups by issuing an $\LCP$-query for each pair of subsequent runs in the list.
\end{enumerate}

We use \cref{prp:runs} to compute an $\Oh(n/\log_\sigma n)$-sized representation of all the runs in~$T$. For runs with small periods, we use the aforementioned approach combined with tabulation (for runs that are not $\tau$-runs) and the following fact for $\tau$-runs.

\begin{fact}[{\cite[Section 6.1.2]{DBLP:conf/stoc/KempaK19},\cite[Lemma 10]{DBLP:conf/esa/Charalampopoulos21}}]\label{fct:tauruns2}
If $\tau\leq \frac{1}{9} \log_{\sigma} n$, given a packed representation of $T$, all $\tau$-runs in $T$ can be sorted by their Lyndon roots in $\Oh(n/\tau)$ time.
\end{fact}

\begin{restatable}{lemma}{shpergroup}\label{lem:shper_group}
All runs in $T$ with periods at most $\tau$, for a given $\tau\le\frac{1}{9} \log_{\sigma} n$, can be grouped by equal Lyndon roots in $\Oh(n/\log_\sigma n)$ time. Among possibly many runs corresponding to equal substrings, at least one needs to be reported, but not necessarily all.
\end{restatable}
\begin{proof}
We compute all such runs using \cref{lem:tauruns,lem:not_tauruns} in $\Oh(n/\log_\sigma n)$ time.

For each cluster of runs $(\R,D)$ returned by the algorithm underlying \cref{lem:not_tauruns}, it suffices to keep one copy of each run in $\R$, say for shift $d=0$.
By \cref{lem:not_tauruns}, this results in just $\Ohtilde(n^{7/18})$ such runs.
For each such run, we compute its Lyndon root and its Lyndon position directly from the definition in $\Ohtilde(1)$ time.
Then the runs can be sorted by Lyndon roots using merge sort in $\Ohtilde(n^{7/18})$ time.

The Lyndon positions of $\tau$-runs can be computed in $\Oh(n/\log_\sigma n)$ time using \cref{fct:tauruns2}. Then we  can merge the two lists of runs sorted by Lyndon roots using longest common extension queries; each of those queries can be performed in $\cO(1)$ time after an $\cO(n/ \log_\sigma n)$-time preprocessing~\cite{DBLP:conf/stoc/KempaK19}.
Using longest common extension queries between consecutive runs in the obtained sorted list, we can also group the runs by Lyndon roots within $\cO(n/ \log_\sigma n)$ time.
\end{proof}

One issue with adapting the aforementioned approach to grouping runs with large periods is that we do not know how to compute the Lyndon positions of $\Oh(n/\log_\sigma n)$ runs in $T$ if their period is greater than $c \log_\sigma n$ for a constant $c$, in $\Oh(n/\log_\sigma n)$ time.

Using cyclic equivalence queries of Kociumaka et al.~\cite{DBLP:journals/siamcomp/KociumakaRRW24} that allow to check if two substrings of $T$ are cyclic rotations of each other, we can check if two runs have the same Lyndon root in $\Oh(1)$ time after $\Oh(n/\log_\sigma n)$-time preprocessing, but this is not sufficient for grouping the runs by Lyndon roots. 
Moreover, it is unknown whether minimal cyclic rotation queries can be implemented efficiently in the packed model.
The fastest known solution, by Kociumaka~\cite{DBLP:conf/cpm/Kociumaka16}, answers minimal cyclic rotation queries in $\Oh(1)$ time but requires $\Theta(n)$ preprocessing; improving the preprocessing time to sublinear here seems to be challenging.

Instead, we use a string synchronizing set (see \cref{fig:sLyn}), as defined by Kempa and Kociumaka~\cite{DBLP:conf/stoc/KempaK19}, to select a unique position within each long-period run (that might not be the Lyndon position) in a consistent way that allows us to group such runs by Lyndon roots.

\begin{definition}[{Synchronizing set~\cite{DBLP:conf/stoc/KempaK19}}]\label{def:sync_sets}
For a length-$n$ string $T$ and a positive integer $\tau \leq \frac12n$, a set $\Sync\subseteq [0\dd  n-2\tau]$ is a \emph{$\tau$-synchronizing set} of $T$ if it satisfies the following two conditions:

\begin{enumerate}
    \item\label{item1} Consistency: If $T[i\dd i+2\tau)=T[j\dd j+2\tau)$, then $i\in \Sync$ if and only if $j\in  \Sync$.

    \smallskip
    \item\label{item2} Density: For $i\in [0\dd  n-3\tau+1]$,\\
    $\Sync \cap [i\dd i+\tau)=\emptyset$ if and only if $\per(T[i\dd i+3\tau -2 ])\leq  \frac13\tau$.
\end{enumerate}
\end{definition}

\begin{remark}
Informally, in the simpler case that $T$ is cube-free, a $\tau$-synchronizing set of~$T$ is an $\Oh(n/\tau)$-sized set of synchronizing positions in $T$ such that each length-$\tau$ fragment of~$T$ (except for the end of the string) contains at least one synchronizing position, and the leftmost synchronizing positions within two length-$3\tau$ matching fragments of $T$ are consistent.
\end{remark}

\begin{figure}[htpb]
\centering
\begin{tikzpicture}[scale=0.4]
      \foreach \i in {0,1,...,31}{
        \draw[xshift=-1cm] (\i+1,0.7) node[above] {\tiny \i};
      }
      \foreach \i/\x in {0/a,1/b,2/a,3/b,4/a,5/a,6/a,7/b,8/a,9/b,10/a,11/a,12/a,13/a,14/b,15/a,16/b,17/a,18/a,19/a,20/b,21/a,22/b,23/a,24/a,25/a,26/a,27/a,28/a,29/a,30/a,31/b}{
        \draw (\i,0) node[above] {$\texttt{\x}$};
      }
    \foreach \x in {2, 5, 8, 10, 12, 15, 18, 21, 22, 26}{
      \draw[xshift=1cm] (\x-1,0) node {\small ---};
    }

      \begin{scope}[xshift=-1cm,thick,yshift=0.3cm]
      \clip (0.5,1) rectangle (13.5,2.5);
      \draw[xshift=0cm] (0.5,1.5) .. controls (2.5,2) and (4.5,2) .. (6.5,1.5);
      \draw[xshift=6cm] (0.5,1.5) .. controls (2.5,2) and (4.5,2) .. (6.5,1.5);
      \draw[xshift=12cm] (0.5,1.5) .. controls (2.5,2) and (4.5,2) .. (6.5,1.5);
      \end{scope}
      \begin{scope}[xshift=10cm,thick]
      \clip (0.5,-1) rectangle (15.5,0);
      \draw[xshift=0cm] (0.5,-0.2) .. controls (2.5,-0.7) and (4.5,-0.7) .. (6.5,-0.2);
      \draw[xshift=6cm] (0.5,-0.2) .. controls (2.5,-0.7) and (4.5,-0.7) .. (6.5,-0.2);
      \draw[xshift=12cm] (0.5,-0.2) .. controls (2.5,-0.7) and (4.5,-0.7) .. (6.5,-0.2);
      \end{scope}
      \draw[xshift=4cm,red,yshift=1.0cm,thick] (0.5,1.5) .. controls (2.5,2) and (4.5,2) .. (6.5,1.5);
      \draw (5,1.5) node {\textcolor{red}{*}};
      \draw[xshift=11cm,red,yshift=-0.7cm,thick] (0.5,-0.2) .. controls (2.5,-0.7) and (4.5,-0.7) .. (6.5,-0.2);
      \draw (12,1.5) node {\textcolor{red}{*}};
      \foreach \x in {22,23,...,29}{
      \draw[blue,xshift=\x cm] (0.5,1.8) .. controls (0.8,2.1) and (1.2,2.1) .. (1.5,1.8);
      }
\end{tikzpicture}
\caption{A string with its $\tau$-synchronizing set $\{2,5,8,10,12,15,18,21,22,26\}$ for $\tau=3$ (underlined letters); note that the lack of a synchronizing position in $[23 \dd 25]$ implies a $\tau$-run (in blue). Two runs with the same Lyndon root \texttt{aaabab} are shown in black. Their sparse-Lyndon positions and (equal) sparse-Lyndon roots \texttt{aababa} are shown in red. In this example, the second run generates a superset of squares generated by the first run.
}\label{fig:sLyn}
\end{figure}
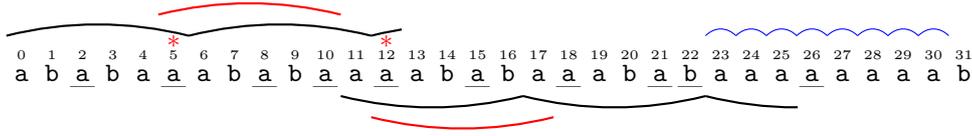

Crucially, string synchronizing sets for small values of $\tau$ can be constructed in optimal time in the packed setting.

\begin{theorem}[{\cite[Proposition 8.10, Theorem 8.11]{DBLP:conf/stoc/KempaK19}}]\label{thm:synch_packed}
For a string $T \in [0\dd \sigma)^n$ with $\sigma=n^{\Oh(1)}$ and $\tau\le \frac15 \log_\sigma n$, there exists a $\tau$-synchronizing set of size $\cO(n/\tau)$ that can be constructed in $\cO(n/\tau)$ time, if $T$ is given in a packed representation.
\end{theorem}

Henceforth, we fix $\tau := \floor{\frac{1}{18} \log_\sigma n}$ and a $\tau$-synchronizing set $\Sync$ for $T$ computed in $\cO(n / \log_\sigma n)$ time using~\cref{thm:synch_packed}.
We next define \emph{sparse-Lyndon positions} (see \cref{fig:sLyn}), noting that their existence is only guaranteed for runs whose periods are long enough, and use them to group runs by Lyndon roots.

\begin{definition}[{Sparse-Lyndon position}]
Position $i$ is a \emph{sparse-Lyndon position} for a periodic fragment $U=T[a \dd b]$ with period $p$ if $T[i \dd n)$ is the lexicographically minimal string among $\{T[j \dd n)\,:\,j \in [a \dd a+p) \cap  \Sync\}$.

If a periodic fragment $U$ with period $p$ has a sparse-Lyndon position $i$, we call $T[i \dd i+p)$ the \emph{sparse-Lyndon root} of $U$ and denote it as $\sLroot(U)$.
\end{definition}

The following key lemma shows that Lyndon positions can indeed be replaced by sparse-Lyndon positions for the sake of grouping runs by Lyndon roots.

\begin{lemma}\label{lem:sparseLyndon}
Both of the following hold.
\begin{enumerate}[(a)]
\item\label{itaa} If a run $R$ has period $p \ge 2\tau-1$, then $R$ has a unique sparse-Lyndon root.
\item\label{itbb} Two runs $R_1$ and $R_2$ with period $p \ge 2\tau$ have the same Lyndon root if and only if they have the same sparse-Lyndon root.
\end{enumerate}
\end{lemma}
\begin{proof}
\eqref{itaa} Let $R=T[a \dd b]$. We will first show that $[a \dd a+p) \cap \Sync$ is non-empty. String $T[a \dd a+p+\tau)$ has period $p$ as $p \ge \tau$. By the periodicity lemma~\cite{fine1965uniqueness}, if $T[a \dd a+p+\tau)$ had a period at most $\frac13\tau$, then the string would have a period $p'$ that is smaller than $p$ and divides~$p$, which is not possible as this would imply that $R$ has a period $p'$. We have $p+\tau \ge 3\tau-1$. String $T[a \dd a+p+\tau)$ contains a fragment of length $3\tau-1$ with period greater than $\frac13\tau$; indeed, otherwise the periods of all such fragments would be equal by the periodicity lemma and this would imply that $T[a \dd a+p+\tau)$ has a period at most~$\frac13\tau$. Let $T[i \dd i+3\tau-2]$ be such a fragment with period greater than $\frac13\tau$. By density, $\Sync \cap [i\dd i+\tau)\ne\emptyset$. We have $a \le i$ and $i+3\tau-2 < a+p+\tau$, so $[i\dd i+\tau) \subseteq [a \dd a+p)$. Hence, indeed, $[a \dd a+p) \cap \Sync \ne \emptyset$.

This shows the existence of a sparse-Lyndon root of $R$.
As no two distinct suffixes are equal, the sparse-Lyndon root of $R$ is unique.

\eqref{itbb} By part \eqref{itaa}, the sparse-Lyndon roots of both runs are well-defined.

The implication ``$\Leftarrow$'' is obvious. As for the implication ``$\Rightarrow$'', let $R_1=T[a \dd b]$ and $R_2=T[a' \dd b']$ and assume that $i$ is the sparse-Lyndon position of $R_1$. By the assumption, there exists $i' \in [a' \dd a'+p)$ such that $T[i' \dd i'+p)=T[i \dd i+p)$. By the fact that $p \ge 2\tau$, we have $T[i' \dd i'+2\tau)=T[i \dd i+2\tau)$, so by consistency, $i' \in \Sync$.

To prove that $i'$ is the sparse-Lyndon position of $R_2$, assume to the contrary that there exists a position $j' \in ([a' \dd a'+p) \cap \Sync) \setminus \{i'\}$ such that $T[j' \dd j'+n) < T[i' \dd i'+n)$. Consequently, $T[j' \dd j'+p) < T[i' \dd i'+p)$ as no two cyclic rotations of a primitive string are equal. By the assumption, there exists $j \in [a \dd a+p)$ such that $T[j \dd j+p)=T[j' \dd j'+p)$. By the fact that $p \ge 2\tau$ and consistency, $j \in \Sync$. We have 
\[T[j \dd j+p)=T[j' \dd j'+p) < T[i' \dd i'+p)=T[i \dd i+p),\text{ so }T[j \dd j+n)<T[i \dd i+n),\]
which contradicts the assumption that $T[i \dd i+p)$ is the sparse-Lyndon root of $R_1$.
\end{proof}

\begin{definition}
The \emph{sparse-Lyndon representation} of a periodic fragment $U$ of $T$ is a quadruple $(\lambda,e,\alpha,\beta)$ such that:

\begin{itemize}
\item $\lambda=\sLroot(U)$, and
\item $U=P\lambda^e S$ with $|P|= \alpha < |\lambda|$ and $|S|=\beta<|\lambda|$.
\end{itemize}
\end{definition}

We use the next lemma to obtain the main result of this section.

\begin{lemma}[{\cite[Theorem 4.3]{DBLP:conf/stoc/KempaK19}}]\label{lem:former_claim}
Given the packed representation of a text $T \in
[0\dd\sigma)^n$ and a $t$-synchronizing set $\S$ of $T$ of size $\Oh(n/t)$ for $t = \Oh(\log_\sigma n)$,
we can compute in $\Oh(n/t)$ time the lexicographic order of all suffixes of $T$ starting at positions in $\S$.
\end{lemma}

\begin{proposition}\label{prp:group}
All runs in $T$ computed as in \cref{prp:runs}, except for the regular layers of \flocks, can be grouped by equal Lyndon roots in $\Oh(n/\log_\sigma n)$ time. For runs with period at most $2\floor{\frac{1}{18} \log_{\sigma} n}$, we compute their Lyndon representations, and for the remaining runs, we compute their sparse-Lyndon representations.
\end{proposition}
\begin{proof}
Recall that $\tau = \floor{\frac{1}{18} \log_\sigma n}$. Runs with periods at most $2\tau$ are grouped by their Lyndon roots using \cref{lem:shper_group}. The remaining runs are grouped by their sparse-Lyndon roots, and thus by Lyndon roots due to \cref{lem:sparseLyndon}, using \cref{lem:former_claim} as follows.

\newcommand{\SRANK}{\mathsf{SparseRANK}}
Let $\Sync=\{s_1,\ldots,s_{|\Sync|}\}$, with $s_1< \dots <s_{|\Sync|}$, be a $\tau$-synchronizing set of $T$ constructed as in \cref{thm:synch_packed}. By \cref{lem:former_claim}, in $\Oh(n/\log_\sigma n)$ time we can construct an array $\SRANK[1 \dd |\Sync|]$ (``sparse RANK'' array) such that \[\SRANK[i]=|\{j \in [1 \dd |\Sync|]\,:\,T[s_j \dd n) \le T[s_i \dd n)\}|.\]
Then, in $\Oh(|\SRANK|)$ time,
we construct a data structure that can answer range minimum queries over $\SRANK$ in $\Oh(1)$ time~\cite{DBLP:conf/latin/BenderF00}.

Let $s_0=-1$ and $s_{|\Sync|+1}=n$ be sentinels.
Let $\Pi$ denote the set of all runs with period $p \ge 2\tau$ that are not regular layers of any pyramid. By \cref{fct:ESA19}, set $\Pi$ can be computed in $\Oh(n/\log_\sigma n)$ time.
For each run $T[a \dd b] \in \Pi$, we need to compute an interval $[u \dd v]$ such that $s_{u-1} < a \le s_u$ and $s_v < a+p \le s_{v+1}$.
By \cref{lem:sparseLyndon}\eqref{ita}, this interval is not empty and hence $u \le v$.
The sparse-Lyndon position of each such run can then be computed in $\Oh(1)$ time as the argmin of a range minimum query over $\SRANK[u \dd v]$.
The positions $u$ and~$v$ are computed for all runs simultaneously in $\Oh(n/\log_\sigma n)$ time by bucket sorting the set
$\{x : x=a \text{ or } x=a+p-1 \text{ for a  run } T[a\dd b] \in \Pi\}$
and merging the obtained sorted list with the synchronizing set $\Sync$ in a merge-sort fashion.

The remainder of the algorithm mimics steps \ref{it2} and \ref{it3}; see the discussion after \cref{obs:sq_from_runs}.
Namely, in $\Oh(n/\log_\sigma n)$ time, we bucket sort the runs with large periods by pairs $(p,\SRANK[i])$, where $i$ is the sparse-Lyndon position and $p$ is the period of the run.
Runs with equal sparse-Lyndon roots form consecutive sublists of the sorted lists.
The equality of sparse-Lyndon roots of consecutive runs in the sorted list can be checked in $\cO(1)$ time using longest common extension queries after an $\Oh(n/\log_\sigma n)$-time preprocessing~\cite[Theorem 5.4]{DBLP:conf/stoc/KempaK19}.
Thus, the grouping is performed in $\Oh(n/\log_\sigma n)$ time.
\end{proof}

\section{Squares Generated by Pyramids}\label{sec:anatomy}
We show that a special square (that is, a square with a primitive and highly periodic half) is always generated by a layer of a pyramid. The proof of the lemma uses the assumption of at least 4 occurrences of the period in a special square half.

\begin{lemma}\label{lem:specialsq_pyramids}
Let $U^2$ be a fragment of $T$. Then $U^2$ is a special square if and only if there exists a pyramid $\Pyr(F,F')$ in $T$ and a layer $R$ such that $U^2 \in 
\fragsquares(R \cap (F \cup F'))$ and $\per(U)=\per(F)$.
\end{lemma}
\begin{proof}
($\Rightarrow$) Let $U^2$ be a special square fragment of $T$ and $p=\per(U)$. Let $F$ and $F'$ be runs with period $p=\per(U)$ that contain the first and the second half of the considered occurrence of $U^2$ in $T$, respectively. We have $F \ne F'$, as otherwise $U^2$ would have period $p$ and, by the periodicity lemma, $U$ would not be primitive.

By \cref{obs:square_run}, there exists a run $R$ in $T$ such that $U^2 \in \fragsquares(R)$. By definition, we have $4p<|U|=\per(R)$. Moreover, $R$ is a subperiodic run with $\per(R)= |U|$ and $\subper(R) \le p$. If we had $\subper(R)=p'<p$, then there would exist a run $G$ in $T$ with period $p'$ that overlaps $F$ or $F'$---say, $F$---on at least $|U|/2$ positions. The overlap length would be greater than $p+p'$, so by the periodicity lemma, the overlap would have period $q:=\gcd(p,p')<p$ that divides~$p$, so $F$ would have period $q$; a contradiction.

Clearly, the runs $F,F'$ are neighboring. We have  $R \in \Pyr(F,F')$ or $R$ is a max-layer of some  pyramid $\Pyr(F'',F''')$ with  $\subper(R) <p$. In either case, $U^2 \subseteq F \cup F'$ and $U^2 \in \fragsquares(R \cap (F \cup F'))$, as required. 

($\Leftarrow$) Let $R$ be a layer of some pyramid with \[\per(F)=\per(F')=\per(U)=p\ \text{and}\ 
U^2 \in \fragsquares(R \cap (F \cup F')).\]
We have $|U| \ge \per(R) > 4p$ since $R$ is subperiodic. By \cref{lem:util},  one half of $U^2$ is contained in $F$ and the other in $F'$.

Period $p$ does not divide $|U|$ as otherwise we would have $F=F'$. Moreover, by the periodicity lemma, $U$ does not have a period $q$ that would divide $U$. Thus, $U$ is primitive and highly periodic, which means that $U^2$ is a special square.
\end{proof}

\renewcommand{\ov}{\mathit{ov}}
\begin{definition}[Pyramid type]
Let $F,F'$ be neighboring runs with period $p$ in $T$. We define the \emph{type} of the pyramid $\Pyr(F,F')$ as a triad $\type(F,F')=(\ov,X,Y)$ where (see \cref{fig:col}):
\[\ov=|F \cap F'|,\quad X=F[|F|-p \dd |F|),\quad Y=F'[0 \dd p).\]
\end{definition}

\begin{remark}
The strings
$X$ and $Y$ are cyclically equivalent if $\Pyr(F,F')$ is non-empty.
\end{remark}

\begin{example} Let $T=T[0\dd 60]=(\mathtt{aaaab})^5\mathtt{a}(\mathtt{aaaab})^7$, $F=(\mathtt{aaaab})^{5}\mathtt{aaaa}=T[0\dd 28]$ and $F'=(\mathtt{aaaab})^7=T[26\dd 60]$. Then 
$\type(F,F')\,=\, (3, \mathtt{baaaa}, \mathtt{aaaab})$.
\end{example}

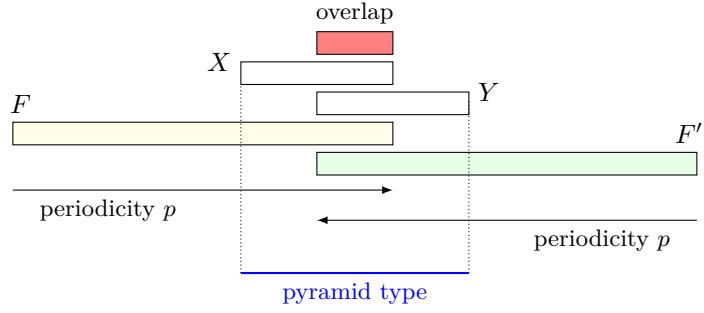
\begin{figure}[htpb]
    \centering
    \input{_fig_col}
    \caption{Illustration of $\type(F,F')=(\ov,X,Y)$, for two runs 
    $F,F'$ with the same period $p$. We have $|X|=|Y|=p$. 
    }\label{fig:col}
\end{figure}

We extend the notation $\fragsquares$ to pyramids as follows.
\begin{definition}[Special squares generated by pyramids]
\[\fragsquares(\Pyr(F,F'))\,:=\,\bigcup_{R\in \Pyr(F,F')}\; \fragsquares(R\cap\, (F\cup F')).\]
We say that the elements of $\fragsquares(\Pyr(F,F'))$ are \emph{generated} by the pyramid $\Pyr(F,F')$.
\end{definition}

\begin{lemma}\label{lem:key}
The sets of squares generated by two pyramids of different types are disjoint.
\end{lemma}
\begin{proof}
Let $F_1,F'_1$ and $F_2,F'_2$ be pairs of neighboring runs with equal periods such that $\type(F_1,F'_1) \ne \type(F_2,F'_2)$.
We will show that the sets $\fragsquares(\Pyr(F_1,F'_1))$ and \linebreak$\fragsquares(\Pyr(F_2,F_2'))$ are disjoint.

Assume there exists $U^2  \in \fragsquares(\Pyr(F_1,F'_1)) \cap \fragsquares(\Pyr(F_2,F'_2))$. We have
\[U^2 \in \fragsquares(R_1 \cap (F_1\cup F_1')) \cap \, \fragsquares(R_2 \cap (F_2\cup F_2')),\] 
for some runs $R_1\in \Pyr(F_1,F_1')$ and $R_2\in \Pyr(F_2,F_2')$.
Let $\type(F_1,F'_1)=(\ov_1,X_1,Y_1)$ and $\type(F_2,F'_2)=(\ov_2,X_2,Y_2)$. 
By \cref{lem:specialsq_pyramids}, $U^2$ is a special square with
$\per(U)=\per(F_1)=\per(F'_1)=\per(F_2)=\per(F'_2).$
 Let $p=\per(U)$. 
 Square $U^2$ does not have period $p$ (as $U$ is primitive). 
Hence, we can define
\begin{itemize}
    \item 
$i$ as the smallest position in $U^2$ such that $\per(U^2[0 \dd i])>p$;
\item 
$j$ as the largest position in $U^2$ such that $\per(U^2[j \dd |U^2|))>p$.
\end{itemize} 
We have $j < |U| \le i$. 
Then, we have \[X_1=U^2[i-p \dd i-1]=X_2,\ Y_1=U^2[j+1 \dd j+p]=Y_2,\ \text{and}\ \ov_1=i-j-1=\ov_2,\] so $\type(F_1,F'_1)=\type(F_2,F'_2)$. This contradiction concludes the proof.
\end{proof}

By $\fragsquares(\RF(F,F'))$ we denote the set of (special) squares generated by regular layers of $\RF(F,F')$.
For an illustration of the next lemma, see \cref{collider}.

\begin{figure}[htpb]
    \centering\includegraphics[width=12cm]{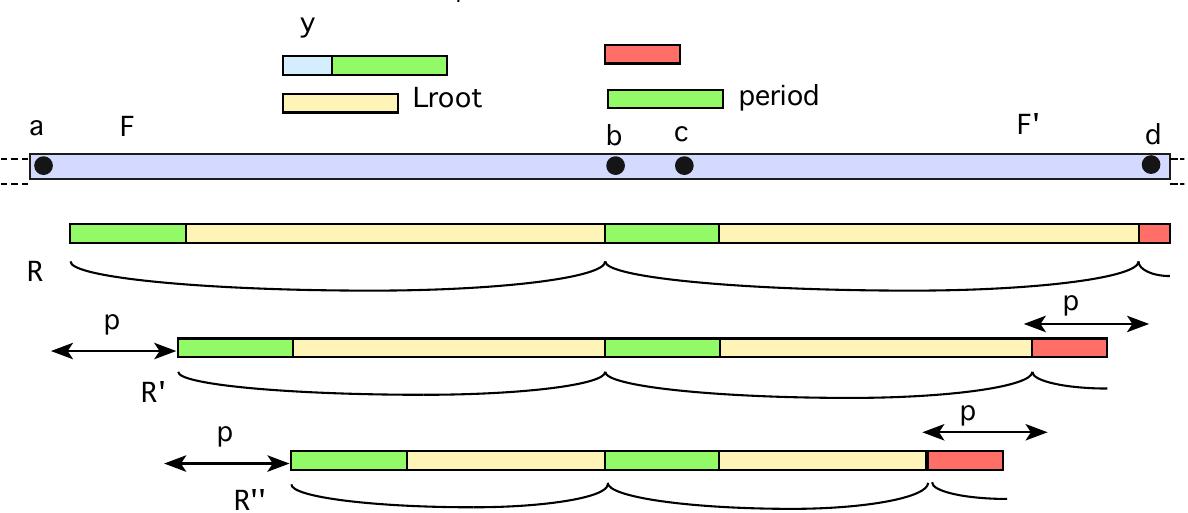}
    \caption{Illustration for the squares generated by the regular layers of a pyramid.
    Consider two neighboring runs $F=T[a\dd c]$ and $F'=T[b\dd d]$ with the same period $p$.
    $\Pyr(F,F')$ consists of the three shown layers $R,R',R''$ with periods at least $4p$ and subperiod $p$; the two shortest layers are regular.
    Let the starting positions of $R'$ and $R''$ in $T$ be $i'$ and $i''$, respectively.
    These positions 
    correspond to occurrences of the factor $Y=F'[0 \dd p)$, shown in green.
    The offsets $b-i'$ and $b-i''$ are equivalent modulo $p$ and are equal to $\per(R')$ and $\per(R'')$, respectively.
    We have that
    $\fragsquares(R') = \{T[i'+x \dd i' + 2(b-i')+x) : x \in [0\dd \ov] \}$
    and
    $\fragsquares(R'') = \{T[i''+x \dd i'' + 2(b-i'')+x) : x \in [0\dd \ov] \}$
    where $\ov=c-b+1$.
    }\label{collider}
\end{figure}

\begin{lemma}\label{lem:spare_key}
Both of the following hold:
\begin{enumerate}[(a)]
\item\label{itb} If $\Pyr(F,F')$ is a pyramid, then \[|\fragsquares(\RF(F,F'))|=|\RF(F,F')|\cdot (|F\cap F'|+1).\]
\item\label{ita} If $\Pyr(F_1,F_1')$ and $\Pyr(F_2,F'_2)$ are pyramids such that $\type(F_1,F'_1)=\type(F_2,F'_2)$, then \[|\RF(F_1,F'_1)| < |\RF(F_2,F'_2)|\ \ \Rightarrow\ \ 
\fragsquares(\Pyr(F_1,F'_1)) \subset \fragsquares(\Pyr(F_2,F'_2)).\]
\end{enumerate}
\end{lemma}
\begin{proof}
Let $F=T[a\dd b]$, $F'=T[a' \dd b']$ be neighboring runs with period $p$ and $a<a'$.
Due to \cref{lem:maxlayer}, the runs in $\RF(F,F')$ are all layers in the set $\mathcal{R} := \{T[x \dd y] \in \Pyr(F,F'): x,y \in (a \dd b')\}$ defined in \cref{lem:contained}, apart, possibly, from the one with the largest period.

\textbf{Proof of \eqref{itb}.}
Each run $R \in \RF(F,F')$ generates $|R|-2\cdot\per(R)+1$ squares. By \cref{lem:contained}, for run $R=T[a'_k \dd b_k]$, this number of squares equals
\[b_k-a'_k+2-2\cdot\per(R)=b-a'+2+2kp+2\delta-2\cdot\per(R)=b-a'+2=|F\cap F'|+1.\]

\textbf{Proof of \eqref{ita}.}
An application of \cref{lem:contained} to $(F_1,F'_1)$ and for $(F_2,F'_2)$ produces equal runs for subsequent values of $k$ if $\type(F_1,F'_1)=\type(F_2,F'_2)$.
Thus, if $|\RF(F_1,F'_1)| \le |\RF(F_2,F'_2)|$, then for each run in $\RF(F_1,F'_1)$, an equal run is present in $\RF(F_2,F'_2)$.
For the max-layer $R_1$ of $\Pyr(F_1,F'_1)$ and the regular layer $R_2 \in \RF(F_2,F'_2)$ with the same period, we have $\fragsquares(R_1 \cap (F_1 \cup F'_1)) \subseteq \fragsquares(R_2 \cap (F_2 \cup F'_2))$.
Runs in a pyramid have different periods, so they generate disjoint sets of squares. The max-layer of $\Pyr(F_2,F'_2)$ thus generates a special square that is not generated by $\Pyr(F_1,F'_1)$.
\end{proof}

\section{Counting Squares}\label{sec:counting}
For brevity, primitively rooted squares are called \emph{p-squares} and non-primitively rooted squares are called \emph{np-squares} (see \cite{DBLP:journals/ejc/KubicaRRW13}).
We note that special squares are, in particular, p-squares.

\subsection{Counting Plain Squares}\label{subsec:counting1}
Recall that a square is plain unless it is special; that is, $U^2$ is plain if it is an np-square or it is not highly periodic.
The next lemma follows from \cite{DBLP:journals/tcs/CrochemoreIKRRW14} (and \cref{fact:bsort}); we give its proof here for completeness.

\begin{restatable}[see {\cite[Theorem 13]{DBLP:journals/tcs/CrochemoreIKRRW14}}]{lem}{lemExtracting}\label{lem:Extracting}
Assume we are given $r$ periodic fragments in $T$ grouped by their Lyndon roots and that the Lyndon representations of all these periodic fragments are available. The numbers of distinct p-squares and distinct np-squares generated by these periodic fragments can be computed in $\Oh(r+\sqrt{n})$ time. 

The same conclusion holds if we are given $r$ periodic fragments in $T$ grouped by their sparse-Lyndon roots and that their sparse-Lyndon representations are available.

In each case, any $k$ distinct corresponding squares can be reported in $\Oh(k+r+\sqrt{n})$ time.
\end{restatable}
\newcommand{\F}{\mathcal{F}}
\begin{proof}
First we show how to count distinct np-squares. For each group of periodic fragments $\F$ with Lyndon root $\lambda$, we compute $\beta=\max_{F \in \F} \floor{|F|/(2p)}$. By \cite[Lemma 22]{DBLP:journals/tcs/CrochemoreIKRRW14}, the group generates all cyclic rotations of np-squares $\lambda^{2\alpha}$, for all $2 \le \alpha < \beta$. This way we obtain $|\lambda|\cdot (\beta-1)$ distinct squares.
Each periodic fragment $F \in \F$ such that $\floor{|F|/(2p)}=\beta$ generates all cyclic rotations $\rot_c(\lambda^{2\beta})$, for $c$ in a union of at most two subintervals of $[0 \dd \lambda)$.
Thus if $\beta>1$, the problem reduces to computing the size of the union of at most $|\F|$ intervals.
This problem can be solved by sorting the endpoints of the intervals and processing them in an increasing order. The sorting can be performed off-line for all groups of fragments using bucket sort in $\Oh(r+\sqrt{n})$ time, as we can treat a number in $[0 \dd n)$ as a pair of numbers in $[0 \dd \floor{\sqrt{n}})$.

The same algorithm can be used for $\lambda$ being the sparse-Lyndon root.

We proceed to counting p-squares. For a group of periodic fragments, if $\beta>1$, we simply have $|\lambda|$ p-squares being rotations of $\lambda^2$; and if $\beta=1$, we reduce the problem to union of intervals (corresponding to rotations of $\lambda^{2\beta}=\lambda^2$) as in the case of np-squares.
\end{proof}

\begin{lemma}\label{lem:np}
    The number of np-squares in $T$ can be computed in $\Oh(n/\log_\sigma n)$ time.
\end{lemma}
\begin{proof}
We use \cref{lem:Extracting} for counting np-squares generated by all runs that are not regular layers, grouped as in \cref{prp:group}. For runs with small periods, we use Lyndon representations, and for the remaining runs we use sparse-Lyndon representations. There are $\Oh(n/\log_\sigma n)$ such runs, so np-squares are counted in $\Oh(n/\log_\sigma n)$ time.
\end{proof}

\begin{lemma}\label{lem:pplain}
    The number of plain p-squares in $T$ can be computed in $\Oh(n/\log_\sigma n)$ time.
\end{lemma}
\begin{proof}
By \cref{lem:specialsq_pyramids}, plain p-squares are generated by runs grouped as in \cref{prp:group}.
For each run computed in \cref{prp:group}, we check if it is also reported as a max-layer in \cref{prp:runs}. This can be checked globally for all runs in $\Oh(n/\log_\sigma n)$ time using bucket sort.
The runs that turned out to be max-layers are cut into smaller periodic fragments that generate plain p-squares (to avoid counting of special squares) as shown below.

Consider a max-\bird $R = T(x \dd y)$ with $\subper(R)=p$.
Let $R_0, \ldots, R_g$ be the sequence of runs with period $p$, sorted with respect to their starting positions, such that $R$ is a max-layer of $\Pyr(R_i,R_{i+1})$ for all $i \in [0 \dd g)$.
Further, let $R_i = T[x_i \dd y_i]$ for each $i \in [0 \dd g]$. 

For convenience, for all $i \in \mathbb{Z} \setminus [0 \dd g]$, set $x_i=\infty$ and $y_i=-\infty$.
For $i \in [0 \dd g]$, let us denote $Y_i := T(\max\{x, y_i - \per(R) +1\} \dd \min\{x_{i+2} + \per(R) - 1, y\})$.
Due to the periodicity of $R$, for any $i,j \in [1\dd g-3]$, we have $Y_i=Y_j$.
Additionally, any occurrence of a plain p-square generated by $R$ in $R$ is contained in some $Y_i$.
Further, each $Y_i$ does not generate any special square of length $2\cdot \per(R)$, as it only contains a single maximal periodic fragment with period $\subper(R)$ that is of length at least $\per(R)$.

Hence, it suffices to use the strings among $Y_0$, $Y_1$, and $Y_{g-2}$ that are of length at least $2\cdot \per(R)$ instead of $R$.
Those strings are periodic and their (sparse-) Lyndon representations can be inferred in $\Oh(1)$ time from the (sparse-) Lyndon representation of the max-layer~$R$.

We obtain $\Oh(n/\log_\sigma n)$ runs that are not max-layers from \cref{prp:group} and  $\Oh(n/\log_\sigma n)$ periodic fragments constructed as described above from max-layers ($\Oh(1)$ periodic fragments from each max-layer).
By \cref{lem:Extracting}, plain p-squares can be counted in $\Oh(n/\log_\sigma n)$ time.
\end{proof}

\subsection{Counting Special Squares}\label{subsec:counting2}
By \cref{lem:specialsq_pyramids}, special squares are only generated by layers of pyramids.

\begin{lemma}\label{lem:group_pyr}
All pyramids $\Pyr(F,F')$ can be grouped by their types in $\Oh(n/\log_\sigma n)$ time.
\end{lemma}
\begin{proof}
Recall that the type of a pyramid $\Pyr(F,F')$ is $\type(F,F')=(ov,X,Y)$ where $ov=|F \cap F'|$, $X$ is a length-$p$ suffix of run $F$, $Y$ is a length-$p$ prefix of run $F'$ and $p=\per(F)=\per(F')$.
By \cref{prp:group}, if $p \le 2\floor{\frac{1}{18} \log_\sigma n}$, we know the Lyndon roots of $F,F'$, and otherwise, we know their sparse-Lyndon roots.

The Lyndon roots of $F$ and $F'$ are the same. We have $X=\rot_{c_X}(\lambda)$ and $Y=\rot_{c_Y}(\lambda)$ for the common Lyndon root $\lambda$ and some values $c_X,c_Y$ that can be computed from the Lyndon representations of $F,F'$ in $\Oh(1)$ time. Instead of grouping pyramids by triads $(ov,X,Y)$, it suffices to group them by quadruples $(\lambda,ov,c_X,c_Y)$. Grouping by Lyndon roots $\lambda$ is performed in \cref{prp:group}. The remaining elements of quadruples are integers in $[0 \dd n)$, so we can bucket sort the quadruples by them in $\Oh(n/\log_\sigma n)$ time in a stable way (so that we do not break the grouping by Lyndon roots) using \cref{fact:bsort}.

The same argument, with sparse-Lyndon roots instead of Lyndon roots, applies for grouping pyramids by types in case the period of runs $F$, $F'$ is greater than $2\floor{\frac{1}{18} \log_\sigma n}$.
\end{proof}

\begin{lemma}\label{lem:pspecial}
    The number of special squares in $T$ can be computed in $\Oh(n/\log_\sigma n)$ time.
\end{lemma}
\begin{proof}
We group the pyramids by their types using \cref{lem:group_pyr}. By \cref{lem:key}, special squares generated by layers from each group can be considered separately. By \cref{lem:spare_key}\eqref{ita}, we can remove all pyramids of the same type that are not of maximal size (in terms of the number of layers).
Among the remaining pyramids, all special squares generated by regular layers are counted using \cref{lem:spare_key}\eqref{itb}.
Special squares generated by max-layers are counted by partitioning each max-layer into periodic fragments generating only special squares as in the proof of \cref{lem:pplain} and then counting all squares generated by such periodic fragments using \cref{lem:Extracting}.
\end{proof}

A combination of \cref{lem:np,lem:pplain,lem:pspecial} implies our main result that we restate for convenience.

\mainthm*

\subsection{Further Results}
Any $k$ squares, for positive $k$ up to the number of distinct squares in $T$, can be listed as follows. \cref{lem:Extracting} allows to report subsequent squares. For special squares, we list squares generated by subsequent runs in the tallest pyramid of each type.
\begin{theorem}
    Given a string $T$ of length $n$ over alphabet $[0\dd \sigma)$ in packed form and integer $1 < k\le |\squares(T)|$, we can output $k$ distinct squares in $T$ in 
    $\Oh(n/\log_\sigma n+k)$ time.
\end{theorem}

Our algorithms generalize to powers with any given exponent $t > 2$ in the same time complexity. In this case, we do not need to consider regular layers, as they generate no powers of exponent greater than 2. Thus an analogue of \cref{lem:Extracting} suffices.
\begin{theorem}
Given a string $T$ of length $n$ over alphabet $[0\dd \sigma)$ in packed form and integer $t>2$, we can compute in $\Oh(n/\log_\sigma n)$ time
    the number of distinct $t$-th powers in~$T$.
\end{theorem}

\section{Proof of \cref{fct:ESA19}}\label{app:fctESA}

We use the \pillar model~\cite{DBLP:conf/focs/Charalampopoulos20} that relies on the following primitive operations. The argument strings of the primitives are fragments of strings in a given collection~$\mathcal{X}$.
\smallskip
\begin{itemize}
\item $\mathsf{Extract}(S, \ell, r)$: Retrieve string $S[\ell\dd r)$.
\smallskip
\item $\LCP(X, Y),\, \LCPR(X, Y)$: Compute the length of the longest common prefix/suffix of $X$ and $Y$.
\smallskip
\item $\IPM(X, Y)$: Assuming that $|Y| < 2|X|$, compute the starting positions of all exact occurrences of $X$ in $Y$, expressed as an arithmetic progression with difference $\per(X)$.
\smallskip
\item $\mathsf{Access}(S, i)$: Retrieve the character $S[i]$;
\smallskip
\item $\mathsf{Length}(S)$: Compute the length $|S|$ of the string $S$.
\end{itemize}
\smallskip
The running time of algorithms in this model can be expressed in terms of the number of primitive \pillar operations (and additional operations not performed on the strings themselves).
The \pillar model admits an optimal implementation in the packed setting.

\begin{fact}[\cite{DBLP:conf/stoc/KempaK19,DBLP:journals/siamcomp/KociumakaRRW24}]\label{fact:pillar_packed}
A string of length $n$ given in a packed representation of size $\cO(n/\log_\sigma n)$ can be preprocessed in $\cO(n/\log_\sigma n)$ time so that any \pillar operation on its substrings can be performed in $\cO(1)$ time.
\end{fact}

A 2-period query for a string $X$, checks if $X$ is periodic and, if so, returns its smallest period. Such a query can be used to compute the period of a given run if it is unknown.

\begin{fact}[{\cite[Corollary 1.10]{DBLP:journals/siamcomp/KociumakaRRW24}}]\label{fct:2per}
A 2-period query can be answered in $\Oh(1)$ time in the \pillar model.
\end{fact}

For a pair of integers $i\ge 0,\,p>0$, we denote by $\gamma(i,p)$
the run with period $p$ containing fragment $T[i\dd i+p)$, if it exists.
\begin{fact}\label{fact:findrun}
    Given a fragment $F= T[i\dd i+p)$ of string $T$, in $\cO(1)$ time in the \pillar model, one can decide if
    $\gamma(i,p)$ exists, and if so, return its value.
\end{fact}
\begin{proof} Let $F' = T[i-\ell \dd i+p+r)$, where
$r = \LCP(T[i\dd |T|), T[i+p\dd |T|))$ and
$\ell = \LCPR(T[0 \dd i+p], T[0\dd i])$.
It then suffices to check if $2p \le |F'|$ and if $F'$ does not have a smaller period.
The latter can be done in $\cO(1)$ time in the \pillar model with \cref{fct:2per}.
\end{proof}

Recall that a canonical representation of $\Pyr(F,F')$ consists of the (endpoints of) runs $F$ and~$F'$, its max-layer, and arithmetic progressions with difference $\per(F)$ specifying the starting positions, ending positions, and periods of its regular layers.
We next show that such a canonical representation can be efficiently computed in the \pillar model for any non-empty pyramid.
This next algorithmic lemma complements the combinatorial \cref{lem:contained,lem:maxlayer}.

\begin{lemma}\label{lem:pyr_can}
Given neighboring runs $F$ and $F'$ of a string $T$ such that $\Pyr(F,F')$ is non-empty, we can compute a canonical representation of $\Pyr(F,F')$ in $\cO(1)$ time in the \pillar model.
\end{lemma}
\begin{proof}
    We use the notation from the aforementioned combinatorial lemmas. That is, $F=T[a\dd b]$, $F'=T[a'\dd b']$, and $p= \per(F)$, assuming without loss of generality that $a < a'$.
    The Lyndon positions of $F$ and $F'$ are $\ell$ and $\ell'$, respectively, and $\delta = (\ell' - \ell) \bmod p$.
    For each $k\in \mathbb{Z}$, we denote $a'_k := a' - k\cdot p - \delta$ and $b_k = b+k\cdot p +\delta$.

    The integer $\delta$ can be computed in $\cO(1)$ time from the output of query $\IPM(T[b-p+1 \dd b], T[a' \dd a'+2p))$.
    Using \cref{lem:contained}, given $\delta$, we compute the endpoints and periods of all runs $T[x \dd y] \in \Pyr(F,F')$ such that $x,y \in (a \dd b')$ as arithmetic progressions with difference $p$ in $\cO(1)$ time.

    For the only non-regular layer, as in the proof of \cref{lem:maxlayer}, it suffices to consider a run containing position $a$ and a run containing position $b'$. For each of them, the period can be determined uniquely in $\Oh(1)$ time based on $\delta$.
    Using \cref{fact:findrun}, we can thus compute the respective run, if it exists, in $\cO(1)$ time in the \pillar model.
    By the lemma, this way at most one run will be computed.

    If \cref{lem:maxlayer} returns a run, we designate it as the max-layer, and return the arithmetic progressions representing $\R$ as $\RF(F,F')$.
    Otherwise, we return the longest run in~$\R$ as a max-layer, and return the arithmetic progressions trimmed by one element as $\RF(F,F')$.
\end{proof}

We filter the explicit runs to avoid double-reporting using the following lemma.

\begin{lemma}\label{claim:filter}
We can check whether a given run $R$ is a regular layer of any pyramid in $\cO(1)$ time in the \pillar model and if so, return the canonical representation of its pyramid.
\end{lemma}
\begin{proof}
We first compute the period of $R= T[a\dd b]$ in $\cO(1)$ time using \cref{fct:2per}.
Then, we similarly check if $R$'s $\per(R)$-length prefix and $\per(R)$-length suffix are also periodic.
If they are not, $R$ cannot be a regular layer of any pyramid (cf.\ \cref{lem:util}), so let us assume that they are.
Let $F$ (respectively, $F'$) be the run containing the $\per(R)$-length prefix of~$R$ (respectively, the $\per(R)$-length suffix of $R$) with period $\per(F)$ (respectively, $\per(F')$) returned by \cref{fact:findrun}.
If $F$ and $F'$ are different, $\per(F) = \per(F')$, and $\per(R) \geq 4\cdot\per(F)$, then~$R$ is a layer-run.
Moreover, we can compute $\Pyr(F,F')$ in $\cO(1)$ time in the \pillar model using \cref{lem:pyr_can} and check whether $R$ is a regular layer.
\end{proof}

We denote $I_k\,=\, \{i \in [0 \dd n)\,:\, i \equiv 0 \pmod{k}\}$.
\begin{observation}\label{obs:anchors}
    If $R=T[a\dd b] \in \RUNS(T)$ and $\per(R) \geq 2k$, then there exists a position $i \in I_k$ such that $i,i+k-1\in [a \dd a+\per(R))$.
\end{observation}

We next prove a version of \cref{fct:ESA19} in the \pillar model, closely following \cite{DBLP:conf/esa/AmirBCK19,Panos}.

\begin{lemma}[see {\cite{DBLP:conf/esa/AmirBCK19,Panos}}]\label{fct:pillarESA19}
Consider a string $T$ of length $n$ and an integer $q \in [1 \dd n]$.
In time $\Oh(n/q)$ in the \pillar model, we can compute a multiset $\X$ of runs such that none of them is a regular layer of any pyramid of $T$
and a multiset $\Y$ of pyramids represented by their canonical representations, such that $|\X|,|\Y|=\Oh(n/q)$, and, for $\Z\,:=\,\bigcup_{(F,F')\in \Y}\, \RF(F,F')$, we have that
\begin{itemize}
\item $\X\cup \Z$
is a superset of all runs in $T$ of period at least $q$, and 
\item $\X \cap \Z=\emptyset$.
\end{itemize}
\end{lemma}

\begin{proof}
We compute the sought runs using a logarithmic number of calls to the algorithm encapsulated by the following claim.

Denote by $\mathbf{R}_x$ the set of all runs in $T$ with period in $[2^x \dd 2^{x+1})$.

\newcommand{\OCC}{\mathit{Occ}}
\begin{claim}\label{at last}
Consider a string $T$ of length $n$ and integer $x \ge 0$.
There exists an algorithm that runs in time $\cO(n/2^x)$ in the \pillar model and returns a representation of all runs in $\mathbf{R}_x$.
The representation consists of $\cO(n/2^x)$ explicitly represented runs (possibly with duplicates), such that none of them is a regular layer of any pyramid, and a collection of $\cO(n/2^x)$ pyramids (possibly with duplicates), each via its canonical representation, whose layers include (among other runs) the remaining runs in $\mathbf{R}_x$.
\end{claim}
\begin{claimproof}
Let $k:=2^{x-1}$ ($k=1$ for $x=0$) and for each $i\in I_k$  denote
\[P_{i} := T[i \dd i+k), \ W_{i} := T[i \dd \min\{i+5k,n\}).\]
Note that $|W_{i}|/|P_i| \leq 5$.
Let $\OCC_i$ be the set of occurrences of $P_{i}$ in $T$ starting at a position in $[i \dd \min\{i+4k+1,n\})$.

 \cref{obs:anchors} implies 
\[\mathbf{R}_x\,=\, \{\,\gamma(i,j-i)\,:\, i\in I_k,\ j\in \OCC_i, \, 2^x \le j-i < 2^{x+1} \ \text{and}\  \gamma(i,j-i)\ \text{exists}\,\}.\]
The set $\OCC_i$ can be as large as $|W_i|$, so we cannot treat each 
$j \in \OCC_i$ separately.
Fortunately $\OCC_i$ consists of only $\cO(1)$ arithmetic progressions, each with difference $\per(P_i)$; $\per(P_i)$ can be computed in $\cO(1)$ time in the \pillar model if $P_i$ is periodic using \cref{fct:2per}, while $\OCC_i$ can be computed using~$\cO(1)$ $\IPM$ queries.
If $P_i$ is periodic, we repeatedly merge any of the computed arithmetic progressions returned by the $\IPM$ queries that contain elements that are $\per(P_i)$ positions apart, that is, maintaining the property that all of our arithmetic progressions have difference $\per(P_i)$.
In the complementary case, that is, when $P_i$ is not periodic, we replace the collection of arithmetic progressions by single elements.

For each $i \in I_k$, we initialize sets $A_i:=\emptyset$ and $B_i:=\emptyset$ that will store explicit runs and pyramids, respectively.

\proofsubparagraph{Processing singleton progressions.}
Consider a trivial progression $\{j\}$ in $\OCC_i$. If $2^x \le j-i < 2^{x+1}$, we use \cref{fact:findrun} to check whether a run $R$ with period $\per(R)= j-i$ exists (then $R=\gamma(i,j-i)$).
We check whether $R$ is a regular layer of any pyramid using \cref{claim:filter}; if not, we insert $R$ to $A_i$, and otherwise we insert to $B_i$ the pyramid of which $R$ is a regular layer, computed using \cref{claim:filter}.

\proofsubparagraph{Processing a single nontrivial arithmetic progression in $\OCC_i$.}
In this case $P_i$ is periodic.
Consider a non-trivial arithmetic progression in $\OCC_i$  
starting with $j$ and ending with~$j'$. We use \cref{fact:findrun} to compute the run $F'= T[a'\dd b']:= \gamma(j,\per(P_i)),$
noting that $\per(F') = \per(P_i)$ due to the periodicity lemma~\cite{fine1965uniqueness}.
We distinguish between two cases: \smallskip
\begin{itemize}
    \item {\bf Case I:} $i\in [a'\dd b']$.\
    In this case, we have $\per(F') = \per(P_i) \leq |P_i| = k$.
    Run $F'$ does not need to be returned as its period does not satisfy the constraints.
    Note that all other elements $j''$ of the considered arithmetic progression are in $[a'\dd b']$ and are contained within~$F'$ and, therefore, they generate run $F'$ and not a run with period $j''-i$.
    Hence, we can discard all the elements of the considered arithmetic progression.
\vskip 0.1cm
    \item {\bf Case II:} $i\not\in [a'\dd b']$.\
    In this case, we first consider $P_i$ as a sample and check whether a run $R$ with period $\per(R) =j-i$ exists.
    If such run exists we report it after checking the constraints as above. Secondly, instead of checking every element of the arithmetic progression  between $j$ and $j'$ (including~$j'$), which would be expensive, we consider run~$F'$ instead.
    We use \cref{fact:findrun} to compute  the run $F = T[a \dd b]:=\gamma(i,\per(P_i))$; see \cref{fig:apr}.
    Note that this run is well-defined as $P_i$ is periodic.
    The computation is then split into two steps.
    \begin{center}
    \begin{figure}[htpb]
        \centering
        \input{_fig_apr19}
        \caption{Processing a 4-elements progression of consecutive occurrences of $P_i$. Here, the shown runs $F=\gamma(i,\per(P_i))=T[a\dd b]$ and $F'=\gamma(j,\per(P_i))=T[a'\dd b']$ are neighbors (we thus have a non-empty pyramid $\Pyr(F,F')$).
        (If $F,F'$ were not neighbors, the considered progression would not imply a pyramid.)
        We have $2^x\le  j-i,\, j'-i < 2^{x+1}$.
        The shown progression implies 
        that the two fragments $T[i\dd i+k)$ and $T[j\dd j'+k)$ are periodic (their common period is $\per(P_i)$).}\label{fig:apr}
    \end{figure}
    \end{center}
    \vspace*{-1.cm}
    \begin{description}
        \item{\bf Step 1.} 
        Here, we compute all runs $R$ that are not subperiodic and are equal to $\gamma(z,y-z)$ for a pair of positions $z\in[a\dd b]$ and  $y\in [a'\dd b']$.
        The non-subperiodicity of such runs means that they are not contained in $F \cup F'$.
        The maximality of runs $F$ and $F'$ then implies that for any such run $R$, either the starting positions of $F$ and $F'$ or the ending positions of $F$ and $F'$ must be contained in the run and be $\per(R)$ positions apart.
        We can then compute such runs, if they exist, using two calls to \cref{fact:findrun}, namely, to compute $\gamma(a,a'-a)$ and $\gamma(b,b'-b)$.
        If these runs exist then for each of these runs that is in $\mathbf{R}_x$, we check whether it is a regular layer of any pyramid using \cref{claim:filter} in $\cO(1)$ time in the \pillar model; if it is, we insert the corresponding pyramid to $B_i$, otherwise we insert the run to $A_i$.
        
        \smallskip
        \item{\bf Step 2.}
        Additionally, if $F$ and $F'$ are neighboring runs, we invoke \cref{lem:pyr_can} and insert the resulting pyramid to $B_i$ if it contains some run in $\mathbf{R}_x$ and the pyramid's max-layer to $A_i$
        if it  is in $\mathbf{R}_x$.
        We can check whether the periodicity constraints are satisfied in $\cO(1)$ time using the canonical representation of the pyramid.
    \end{description}
\end{itemize}

\smallskip\noindent In the algorithm, we might insert the same run or pyramid multiple times. However, the total number of insertions is as required.
The final result is a pair of two multisets: the union of all computed $A_i$,
and the union of all $B_i$, over $i\in I_k$.

\proofsubparagraph{Time complexity.}
    For each $i\in I_k$, using a constant number of $\IPM$ queries, we can compute a representation of $\OCC_i$ as a constant number of arithmetic progressions, each with difference $\per(P_i)$.
    We then spend $\cO(1)$ time in the \pillar model to process $P_i$ and the set $\OCC_i$ of its occurrences in $W_i$.
    Over all $i\in I_k$, the performed computations require total time $\cO(n/k) = \cO(n/2^x)$ in the \pillar model.
\end{claimproof}

We call \cref{at last} for each $x \in [\lfloor \log q \rfloor \dd \lfloor \log n \rfloor]$.
This takes total time at most
\[\sum_{x=\lfloor \log q \rfloor}^{\infty} \cO(n/2^x) = \cO(n/q).\]
Finally, we filter the computed sets, in constant time per element, by removing all runs whose period is in $[2^{\lfloor \log q \rfloor} \dd q)$ and all pyramids all of whose regular layers have period at most $q-1$ in constant time per element; we use \cref{fct:2per} for explicitly represented runs.
This concludes the proof of this lemma.
\end{proof}

A combination of the above lemma (with $q=\ceil{c \log_\sigma n}$) with the optimal implementation of the \pillar model in the packed setting (see \cref{fact:pillar_packed}), and \cref{fact:bsort} to remove duplicates, yields \cref{fct:ESA19}, restated here for convenience.

\fctESA*
   
\bibliographystyle{plainurl}
\bibliography{references}

\end{document}

%% file: _fig_pyramid.tex
\begin{tikzpicture}
\pgfmathsetmacro{\textWidth}{0.3}
\pgfmathsetmacro{\H}{0.3}
\foreach \i in {0, ..., 31} {
  \pgfmathsetmacro{\character}{%
    ifthenelse(\i<16,
      ifthenelse(mod(\i,2)==0,"a","b"),
      ifthenelse(mod(\i,2)==0,"b","a")
    )
  }
  \pgfmathsetmacro{\col}{ifthenelse(\i<16,"red!80!black","green!50!black")}
  \node (v\i) at (\i*\textWidth,0) [above,\col] {\small \tt \character};
}
\foreach \x [count=\i] in {1, 3, 5, 7, 9, 11} {
  \pgfmathsetmacro{\y}{31-\x}
  \def\left{(\x*\textWidth-0.4*\textWidth,0.3+\i*\H)}
  \def\right{(\y*\textWidth+0.4*\textWidth,0.3+\i*\H)}
  \draw \left--\right;
  \draw \left--+(0,-\H*0.7);
  \draw \right--+(0,-\H*0.7);
}
\end{tikzpicture}

%% file: _fig_col.tex
\begin{tikzpicture}
\coordinate (ff) at (4,-0.4);
\coordinate (ffr) at ($(ff)+(5, 0.3)$);
\coordinate (f) at (0,0);
\coordinate (fr) at ($(f)+(5, 0.3)$);

\draw[fill=yellow!10!white] (f) rectangle (fr) node[above,xshift=-4.9cm] {$F$};
\draw[fill=green!10!white] (ff) rectangle (ffr) node[above,xshift=-0.1cm] {$F'$};

\draw ($(ff)+(0, 0.8)$) rectangle +(2, 0.3) node[right] {$Y$};
\draw ($(fr)+(0, 0.5)$) rectangle +(-2, 0.3) node[left] {$X$};
\draw[fill=red!50!white] ($(fr)+(0, 0.9)$) rectangle +(-1, 0.3) node[midway,yshift=0.4cm] {\small overlap};

\draw[-latex] ($(f)+(0, -0.6cm)$) -- +(5, 0) node[near start,below] {\small periodicity $p$};
\draw[-latex] ($(ffr)+(0, -0.9cm)$) -- +(-5cm, 0) node[near start,below] {\small periodicity $p$};

\draw[thick,blue] ($(ff)+(-1cm, -1.3cm)$) -- +(3, 0) node[midway,below] {\small pyramid type};

\draw[densely dotted] ($(ff)+(-1cm, -1.3cm)$) -- ($(ff)+(-1cm, 1.2cm)$);
\draw[densely dotted] ($(ff)+(2cm, -1.3cm)$) -- ($(ff)+(2cm, 1cm)$);
\end{tikzpicture}

%% file: _fig_apr19.tex
\begin{tikzpicture}
\tikzset{
    ncbar angle/.initial=90,
    ncbar/.style={
        to path=(\tikztostart)
        -- ($(\tikztostart)!#1!\pgfkeysvalueof{/tikz/ncbar angle}:(\tikztotarget)$)
        -- ($(\tikztotarget)!($(\tikztostart)!#1!\pgfkeysvalueof{/tikz/ncbar angle}:(\tikztotarget)$)!\pgfkeysvalueof{/tikz/ncbar angle}:(\tikztostart)$)
        -- (\tikztotarget)
    },
    ncbar/.default=0.5cm,
}
\tikzset{square right brace/.style={ncbar=-0.5cm}}

\node at (7,1) {\small arithmetic progression};
\draw (0,0) rectangle (10, 0.3);

\node (a) at (0.2,0) [below] {$a$};
\node (aa) at (5.2,0) [below] {$a'$};
\node (i) at (1.0,0.3) [above] {$i$};
\draw[fill=green!30!white] (1.0,0) rectangle +(2, 0.3) node[midway,below,yshift=-0.2cm] {$P_i$};

\node (j) at (5.8,0.3) [above] {$j$};
\node (jj) at (8.8,0) [below] {$j'$};
\node (b) at (6.5,0) [below] {$b$};
\node (bb) at (9.5,0) [below] {$b'$};
\foreach \x in {5.8, 6.8, 7.8, 8.8} {
  \filldraw[fill=red!30!white] (\x,0.15) circle (0.1cm);
}

\coordinate (p) at ($(a)+(0,-0.3cm)$);
\draw (p) to [ncbar=-0.3cm] (b|-p);
\node at ($(a)!0.5!(b)+(0,-0.8cm)$) {$F$};

\coordinate (pp) at ($(aa)+(0,-0.3cm)$);
\draw (pp) to [ncbar=-0.6cm] (bb|-pp);
\node at ($(aa)!0.5!(bb)+(0,-1.1cm)$) {$F'$};

\end{tikzpicture}